\documentclass{article}
\usepackage{orcidlink}
\usepackage{fullpage}
\usepackage{authblk}
\usepackage{amsmath, amsthm, amsfonts,  amssymb, latexsym,mathrsfs}
\usepackage{graphicx} 
\usepackage{bm}
\usepackage{pdflscape}
\usepackage{setspace}

\usepackage{tikz}
\usetikzlibrary{patterns}
\usetikzlibrary{shapes.geometric}
\usetikzlibrary{decorations.pathmorphing}
\usetikzlibrary{arrows.meta}
\tikzset{snake it/.style={decorate, decoration=snake}}
\usetikzlibrary{positioning}
\usetikzlibrary{arrows,shapes,positioning}
\usetikzlibrary{decorations.markings}
\tikzstyle arrowstyle=[scale=1]
\tikzstyle directed=[postaction={decorate,decoration={markings,mark=at position .65 with {\arrow[arrowstyle]{stealth}}}}]
\tikzstyle reverse directed=[postaction={decorate,decoration={markings,mark=at position .65 with {\arrowreversed[arrowstyle]{stealth};}}}]

\newtheorem{proposition}{Proposition}
\newtheorem{theorem}{Theorem}

\newtheorem{corollary}{Corollary}
\newtheorem{lem}{Lemma}

\newcommand{\Z}{\mathbb{Z}}

\newcommand{\cO}{{\mathcal{O}}}

\newcommand{\blambda}{{\mkern0.75mu\mathchar '26\mkern -9.75mu\lambda}}
\newcommand{\ud}{{\rm d}}
\newcommand{\del}{\partial}

\newcommand{\rin}{{\mathrm{in}}}
\newcommand{\rup}{{\mathrm{up}}}
\newcommand{\rout}{{\mathrm{out}}}
\newcommand{\rdown}{{\mathrm{down}}}

\DeclareMathOperator{\Real}{Re}
\DeclareMathOperator{\Imag}{Im}
\DeclareMathOperator{\Tr}{Tr}
\DeclareMathOperator{\diag}{diag}

\newcommand{\defeq}{\mathrel{:=}}

\newcommand{\btheta}{\boldsymbol{\theta}}

\renewcommand{\bm}{\boldsymbol{m}}

\newcommand{\ben}{\begin{equation}}
\newcommand{\een}{\end{equation}}
\renewcommand{\pounds}{{\mathscr L}}

\begin{document}

\author{Stefan Hollands \orcidlink{0000-0001-6627-2808}}
\affil{\small Institute for Theoretical Physics, Leipzig University, Br\" uderstrasse 16, 04103 Leipzig, and MPI-MiS, Inselstrasse 22, 04103, Leipzig, Germany, stefan.hollands@uni-leipzig.de}

\author{Akihiro Ishibashi \orcidlink{0000-0002-3485-9935}}
	\affil{\small Department of Physics and Kobayashi-Maskawa Institute, Nagoya University, Nagoya 464-8602, Japan 
}%

\author{Jochen Zahn}
\affil{\small Institute of Theoretical Physics, Leipzig University, Br\" uderstrasse 16, 04103 Leipzig}

\title{Hidden Symmetry of Kerr-deSitter from Manifest Symmetry of Painlevé VI}

\maketitle

\begin{abstract}
We show that the recently discovered ``mass symmetries" of the radial Teukolsky equation for Kerr-deSitter black holes can be understood 
from a corresponding symmetry of the Painlevé VI equation via the classical theory of ``isomonodromic deformations". It is known that a subset of the 
mass symmetries transforms the Teukolsky equation to a wave equation on a ``dual" Lorentzian spacetime. We find that this dual 
spacetime again represents a black hole, and use this insight to re-obtain a previously known partial mode stability result in a geometric manner. Further special geometric 
features of the dual spacetime are pointed out.
\end{abstract}

\section{Introduction}

The physical significance of the exact Kerr-Newman family of black hole solutions is primarily based on the uniqueness theorems (see e.g., \cite{Chrusciel:2012jk} for a review), and their stability against small perturbations (see e.g., \cite{Hintz:2022lgf} for a review). Although by far not sufficient for a rigorous proof of stability, a key step is to establish that the linearized Einstein-matter equations on the black hole spacetime under consideration have no (marginally) unstable quasinormal modes (QNMs), 
i.e. ones having a frequency with $\Imag \omega > \! (=) \, 0$. 

In Kerr, and also in several of its generalizations, QNMs can be analyzed within the Teukolsky framework of black hole perturbation theory \cite{Teukolsky:1973ha}. 
The first direct search for instability within this framework was carried out by \cite{Press:1973zz}. Considerably later, \cite{Whiting} proved mode stability (no QNMs with $\Imag \omega > 0$) for gravitational, EM and massless scalar perturbations of subextremal Kerr. The absence of QNMs with $\Imag \omega = 0$ was established for various spins by \cite{Shlapentokh-Rothman:2013hza}, \cite{Shlapentokh-Rothman:2020vpj}, \cite{Andersson:2016epf}. 
The extension to the extremal case is due to \cite{TeixeiradaCosta:2019skg}.

It is natural to ask whether mode stability holds in related settings. This is not necessarily the case; for example, one finds instability of the massive Klein-Gordon equation on Kerr for open set of mass parameters \cite{Shlapentokh-Rothman:2013ysa}, \cite{Dolan:2007mj}, or for perturbations of Kerr-Anti-deSitter spacetimes \cite{Carter} violating the Hawking-Reall bound \cite{Dold:2015cqa}, \cite{Green:2015kur}.
On the other hand, \cite{Dyatlov:2010hq} has established mode stability for the minimally coupled wave equation on slowly rotating Kerr-deSitter (Kerr-dS) by perturbative arguments, 
and \cite{Hintz} has established mode stability of the minimally coupled wave equation on Kerr-dS with sufficiently small cosmological constant.
For mode stability for ungauged gravitational perturbations of slowly rotating Kerr-dS see \cite{Hintz:2016gwb}. Going beyond perturbative results, \cite{CTdC} have excluded a 
certain range of unstable QNMs, though their results fall short of proving mode stability.  

In general, the essential technical difficulty in proving non-perturbative results about mode stability is the absence of a manifestly positive energy functional of the perturbations in many cases. 
For gravitational perturbations of Schwarzschild spacetimes, this issue was addressed by \cite{Regge:1957td}, \cite{Zerilli:1970se} via a transformation theory, which however could not be generalized to Kerr. The main reason is that Kerr possesses an ergoregion, giving rise to superradiant modes in the frequency range
\ben
\label{superboundKerr}
0 \le \omega/m \le \omega_+
\een
where $m$ describes the angular oscillation $e^{i m \phi}$ of the mode and where $\omega_+$ is the angular velocity of the event horizon. For frequencies in this window, 
there is a negative energy flux across the event horizon, corresponding to a negative definite contribution in the natural energy even of a massless scalar field perturbation, see e.g., 
\cite{waldbook}.

To get around this issue, \cite{Whiting} introduced a new ``hidden symmetry transformation'' of the radial Teukolsky equation---as well as less surprising differential transformations for the angular Teukolsky equation. The transformation \cite{Whiting} is ``hidden'', in the sense that 
it cannot be obtained by a change of variables or multiplicative renormalization of the Teukolsky equation. The transformed radial equation has the {\it same} spectrum of 
QNMs, and at the same time eliminates the problem with superradiant modes. A generalization of the transformations \cite{Whiting} to Kerr-dS was subsequently found by \cite{Umetsu}.
More recently, \cite{Hatsuda:2020sbn} has suggested a connection of these types of spectral symmetries with ``mass symmetries'', borrowing from a formal connection 
with supersymmetric quantum gauge theories. Even though the mass symmetries are of a similar nature as those 
derived earlier\footnote{\label{Umetsuissue} The final form of the transformed equation by \cite{Umetsu} differs from \cite{CTdC} and our results below. We suspect 
typographical errors in \cite{Umetsu}.} \cite{Umetsu}, this work provided a new perspective. Later, \cite{CTdC} mathematically established the mass symmetries
by a  method based on the Leaver \cite{Leaver:1985ax} and MST \cite{STU98} approaches to spectral problems for the Teukolsky equation.

Both the integral transformations by \cite{Umetsu} or the approach by \cite{CTdC} based on \cite{Leaver:1985ax}, \cite{STU98}  
arguably do not provide a particularly transparent understanding of the origin of the mass symmetries. In this paper, we therefore propose another derivation
that is based on the so-called isomonodromy method. In that method (see \cite{Jimbo:1981tov}, \cite{JimboMiwaII}, \cite{NMLC} as general references), 
one considers certain continuous deformation of  ordinary differential equations with rational coefficients, such as the 
radial Teukolsky equation on Kerr-dS. The deformation is designed so as to leave the monodromies of the equation invariant. Since the monodromies 
can be related to spectral information \cite{NMLC}, the flow can be said to be ``isospectral''. The flow of the coefficients of the differential equation is governed,
in the case of the radial Teukolsky equation studied here, by the Painlevé VI equation, since the radial Teukolsky 
equation on Kerr-dS is a Heun equation with four regular singular points  \cite{STU98}. 

Written in a suitable form \cite{Jimbo}, the Painlevé VI equation \eqref{sigma6}
is manifestly invariant under the mass symmetry transformations. In this paper, we show how to relate this invariance 
to QNMs, and we show that it leads---at least for real frequencies---to the same bounds as obtained by \cite{CTdC}, even though 
the final form of the radial Teukolsky equation after the mass symmetry appears to be somewhat different.

As observed by \cite{Whiting}, if one recombines the radial and angular Teukolsky equations after the symmetry transformations in Kerr, then the resulting equation for the {\it full} mode 
function is a wave equation on a new spacetime without ergoregions. That observation has been generalized to Kerr-dS by \cite{Umetsu}. With the aim of a better 
geometric understanding of the Kerr-dS mass symmetries, we investigate the geometry of this ``dual'' spacetime. Surprisingly to us, it again has the structure of a black hole spacetime
with bifurcate Killing horizons, with ``dual'' surface gravity and angular velocity parameters. 
Contrary to the Kerr case, there are still ergoregions. As we show, unlike the Teukolsky master equation 
on Kerr-dS, the wave equation on the dual spacetime admits a conserved current, and we use this current to give a geometric interpretation of the bound \cite{CTdC} on 
QNM frequencies on Kerr-dS (for real frequencies). It can be stated simply as
\ben
\label{improvesuperr}
\hat \omega_c \le \omega/m \le \hat \omega_+,
\een
where $\hat \omega_c, \hat \omega_+$ refer to the angular velocities of the cosmological- and event horizons in the {\it dual} spacetime. 

By comparison, the superradiant bound in the original Kerr-dS spacetime analogous to \eqref{superboundKerr} would be $\omega_c \le \omega/m \le \omega_+$.
Note that $\hat \omega_+ < \omega_+$ and $\hat \omega_c>\omega_c$, so the bound \eqref{improvesuperr} is sharper than the naive superradiant bound 
obtained by straightforward analysis in the {\it original} spacetime. 
In particular, for $\Lambda \to 0$ or for $a \to 0$, we get $\hat \omega_+, \hat \omega_c \to 0$ by \eqref{homegadef}, so one re-obtains the known results about the absence of 
real QNMs for the Kerr- and Schwarzschild-dS spacetimes.

This paper is organized as follows. In section \ref{sec:background} we recall the Kerr-dS spacetime and the Teukolsky equations for perturbations on Kerr-dS.  
In section \ref{sec:transformationtheory}, we present our transformation theory based on the isomonodromy method. In section \ref{sec:geometry} we investigate the 
geometry of the dual spacetime of Kerr-dS and the connections to the QNM bounds. Some related technical material is relegated to appendix \ref{app:B} and \ref{app:e}. 
In appendix \ref{app:A} we record some non-obvious curious facts that we found out about the geometry of the dual Kerr metric.

\medskip
Our {\bf notations and conventions} agree with those in \cite{waldbook}. Hatted quantities refer to the dual spacetime or transformed parameters under the mass symmetry.
 


\section{Kerr-dS and its linear perturbations}
\label{sec:background}

\subsection{Extended Kerr-dS spacetime}
The Kerr-dS metric \cite{Carter} can be presented in the  form \cite{ChambersMoss}
\begin{equation}
\label{BLKdS}
 \ud s^2 = - \frac{\Delta_r}{\rho^2 (1+\alpha)^2} \left( \ud t - a \sin^2 \theta \ud \phi \right)^2 + \frac{\Delta_\theta}{\rho^2 (1+\alpha)^2} \left[ a \ud t - (r^2 + a^2) \ud \phi \right]^2 + \rho^2 \left( \frac{\sin^2 \theta}{\Delta_\theta} \ud \theta^2 + \frac{\ud r^2}{\Delta_r} \right),
\end{equation}
with
\begin{align}
 \alpha & = a^2/L^2, & \rho^2 & = r^2 + a^2 \cos^2 \theta, \\
 \Delta_\theta & = \sin^2\theta(1 + \alpha \cos^2 \theta), & \Delta_r & = ( r^2 + a^2 ) (1 - r^2 / L^2 ) - 2 M r,
\end{align}
where $a,M,L$ are real parameters and $(\theta,\phi)$ are spherical polar coordinates on a 2-sphere. 
The coordinates $x^\mu=(t,r,\theta,\phi)$ are analogous to Boyer-Lindquist (BL) coordinates for the Kerr spacetime. 

Throughout this paper, we restrict our attention to the ``physical'' range of parameters such that  $\Delta_r$ has three real unequal positive roots, denoted by $0<r_-<r_+<r_c$,
\begin{equation}
\label{Deltardef}
    \Delta_r = -L^{-2}(r-r_+)(r-r_-)(r-r_c)(r-r_o).
\end{equation}
The fourth, negative, root is $r_o = - (r_- + r_+ + r_c)$. 

A priori, the BL coordinates cover the patch $r_+<r<r_c$ of the spacetime only, but this patch and \eqref{BLKdS} can be analytically continued, as described in detail in \cite{Akcay:2010vt, Borthwick:2018qsb}. Precisely in the range $0<r_-<r_+<r_c$, the extended spacetime describes a subextremal rotating black hole in an expanding universe, with specific angular momentum parameter $a$, mass parameter $M$, and Hubble radius $L$. The extended Kerr-dS manifold is ${\mathcal M} = {\mathcal P} \times S^2$, where $\mathcal P$ may be represented by a well-known \cite{Akcay:2010vt} Penrose diagram. For comparison with the ``dual metric'' presented later, we recall this conformal diagram in fig. \ref{fig:0}.

The extended spacetime is a solution to the vacuum Einstein equation with positive cosmological constant $\Lambda$, 
\begin{equation}
R_{ab}=\Lambda g_{ab},
\end{equation}
related to $L$ by $L^2 = \frac{3}{\Lambda}$. In the extended spacetime, the coordinate value $r_-$ corresponds to the inner bifurcate Killing horizon, $r_+$ to the outer (or event-) bifurcate Killing horizon, and $r_c$ corresponds to the cosmological bifurcate Killing horizon. There exist combinations of the $t$-translation Killing vector field (VF) $T^a=(\partial_t)^a$ and the rotation Killing VF $\Phi^a=(\partial_\phi)^a$ which are tangent and normal to the corresponding bifurcate Killing horizons in the extended spacetime,\footnote{Here the sign
$\pm_j$ depends on what portion (e.g., future/past) of the $r=r_j$ horizon one is looking at.}
\begin{equation}
\label{kjdef}
    K_j^a = T^a + \omega_j \Phi^a, 
    \quad K_j^a \nabla_a K_j^b = \pm_j \kappa_j K_j^b, \quad j \in \{+,-,c\}. 
\end{equation}
Here, $\kappa_j>0$ and $\omega_j$ are the surface gravities and angular velocities of the respective horizons, given by
\ben
\label{eq:wk_KdS}
\kappa_j=\frac{|\partial_r\Delta_r|_{r=r_j}}{2(1+\alpha) (r_j^2+a^2)}\, , \quad
 \omega_j = \frac{a}{r_j^2 + a^2}, \quad j \in \{-,+,c\}.
\een
		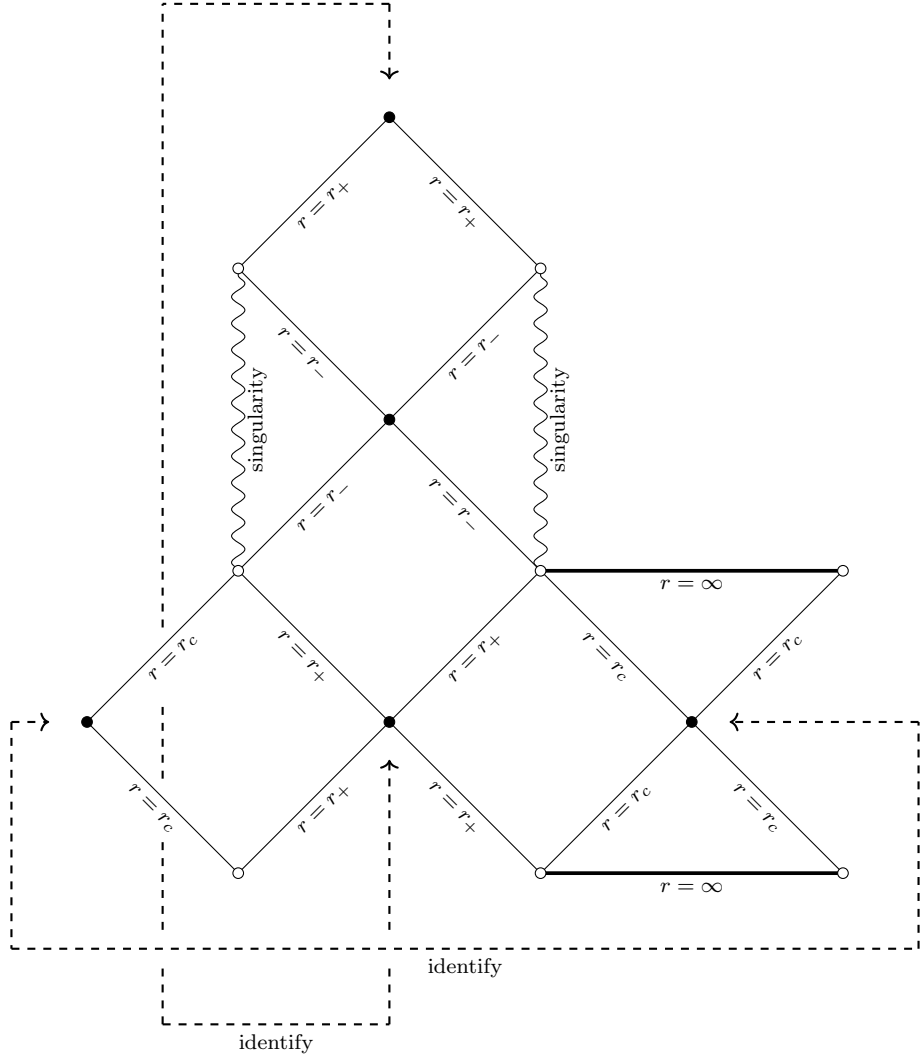
\begin{figure}
		\centering
			\begin{tikzpicture}[scale=1]
\draw (0,4) -- (-2,6) node[midway, below, sloped]{\footnotesize $r=r_-$};
\draw (-2,6) -- (-4,8) node[midway, below, sloped]{\footnotesize $r=r_-$};	
\draw (-2,6) -- (0,8) node[midway, below, sloped]{\footnotesize $r=r_-$};
\draw (-2,6) -- (-4,4) node[midway, below, sloped]{\footnotesize $r=r_-$};		
\draw (0,8) -- (-2,10) node[midway, below, sloped]{\footnotesize $r=r_+$};
\draw (-2,10) -- (-4,8) node[midway, below, sloped]{\footnotesize $r=r_+$};		
\draw (0,0) -- (-2,2) node[midway, below, sloped]{\footnotesize $r=r_+$};
\draw (-2,2) -- (-4,4) node[midway, below, sloped]{\footnotesize $r=r_+$};
\draw (0,0) -- (2,2) node[midway, below, sloped]{\footnotesize $r=r_c$};
\draw (2,2) -- (0,4) node[midway, below, sloped]{\footnotesize $r=r_c$};
\draw (-2,2) -- (0,4) node[midway, below, sloped]{\footnotesize $r=r_+$};
\draw[black, thick, dashed,->] (-7,2) -- (-6.5,2);
\draw[black, thick,dashed, ->] (-2,11.5) -- (-2,10.5);
\draw[black, thick,dashed] (-2,11.5) -- (-5,11.5);

\draw[black, thick,dashed] (-5,-2) -- (-5,-1.25);
\draw[black, thick,dashed] (-5,-0.75) -- (-5,.5);
\draw[black, thick,dashed] (-5,1.25) -- (-5,2.25);
\draw[black, thick,dashed] (-5,3.25) -- (-5,11.5);

\draw[black, thick,dashed] (-5,-2) -- (-2,-2) node[midway, below, black]{\footnotesize identify};
\draw[black, thick,dashed] (-2,-2) -- (-2,-1.25);
\draw[black, thick,dashed,->] (-2,-.75) -- (-2,1.5);
\draw[black, thick,dashed] (-7,2) -- (-7,-1);
\draw[black, thick,dashed] (-7,-1) -- (5,-1) node[midway, below, black]{\footnotesize identify};
\draw[black, thick,dashed] (5,-1) -- (5,2);
\draw[black, thick,dashed,->] (5,2) -- (2.5,2);
\draw[snake it] (0,4) -- (0,8) node[midway, below, sloped]{\footnotesize singularity};
\draw[snake it] (-4,4) -- (-4,8) node[midway, below, sloped]{\footnotesize singularity};
\draw[double=black] (0,4) -- (4,4) node[midway, below, sloped]{\footnotesize $r=\infty$};
\draw[double=black] (0,0) -- (4,0) node[midway, below, sloped]{\footnotesize $r=\infty$};
\draw (2,2) -- (4,4) node[midway, below, sloped]{\footnotesize $r=r_c$};
\draw (2,2) -- (4,0) node[midway, below, sloped]{\footnotesize $r=r_c$};
\draw (-2,2) -- (-4,0) node[midway, below, sloped]{\footnotesize $r=r_+$};
\draw (-4,0) -- (-6,2) node[midway, below, sloped]{\footnotesize $r=r_c$};
\draw (-4,4) -- (-6,2) node[midway, below, sloped]{\footnotesize $r=r_c$};

\draw[fill=white] (-4,4) circle (2pt);
\draw[fill=white] (-4,0) circle (2pt);
\draw[fill=white] (4,4) circle (2pt);
\draw[fill=white] (4,0) circle (2pt);
\draw[fill=white] (0,8) circle (2pt);
\draw[fill=white] (-4,8) circle (2pt);
\draw[fill=white] (0,0) circle (2pt) node[below]{$$};
\draw[fill=black] (-2,2) circle (2pt);
\draw[fill=black] (-6,2) circle (2pt);
\draw[fill=black] (-2,6) circle (2pt);
\draw[fill=black] (-2,10) circle (2pt);
\draw[fill=white] (0,4) circle (2pt) node[above right]{$$};
\draw[fill=black] (2,2) circle (2pt);
\end{tikzpicture}
\caption{Penrose diagram $\mathcal P$ of extended Kerr-dS \cite{Akcay:2010vt}. This spacetime could be further extended through a part of the singularity.}
		\label{fig:0}
		\end{figure}

The following limiting cases are well-known:
\begin{itemize}
    \item $a=0$ ($\Leftrightarrow r_-=0$): Schwarzschild-dS spacetime.
    \item $L=\infty$ ($\Leftrightarrow r_c=\infty$): Kerr spacetime.
    \item $r_+=r_-$: Extremal Kerr-dS spacetime.
    \item $a=M=0$ ($\Leftrightarrow r_+=r_-=0$): dS spacetime.
    \item $r_+=r_c$: rotating Nariai spacetime (in a scaling limit).
\end{itemize}

\subsection{Perturbations}

The Teukolsky master equation for gravitational, electromagnetic, chiral spinor, and massless, conformally coupled scalar field perturbations for $\Lambda>0$ was first obtained by \cite{Khanal:1983vb}. 
As for $\Lambda=0$, the equations can be separated, leading to 
an angular and a radial equation of potential barrier type for each spin, $s \in {\mathbb Z}/2$, through the ansatz
\begin{equation}
   \Psi(t,r,\theta,\phi) = 
    R(r) S(\theta) e^{-i\omega t+im\phi},
\end{equation}
where $m \in \mathbb{Z}$ and $\omega \in \mathbb{C}$.
For gravitational perturbations, $s=\pm 2$, and $\Psi$
is proportional to the top resp. bottom Newman-Penrose component of the linearized Weyl tensor, or alternatively the ingoing or outgoing radiation gauge Hertz potential. 

The spin $s$ radial equation is, see \cite{ChambersMoss}, 
\begin{equation}
\label{eq:Rmaster}
\left( -\Delta_r^{-s} \del_r \Delta_r^{s+1} \del_r + V_r  \right) R = 0,
\end{equation}
where 
\begin{align}
 W_r & = (1+\alpha) \left[ \omega (r^2 + a^2) - a m \right], \\
V_r & = 2 (s+1)(2s+1)r^2/L^2 + \blambda -s(1-\alpha) 
- \Delta_r^{-1} \left( W^2_r - i s W_r \partial_r \Delta_r \right)  -2 i s \partial_r W_r.
\end{align}
Here $\blambda$ is the angular separation constant, for which we adopt the same conventions as \cite{NMLC,Suzuki:1999nn,Umetsu}.\footnote{The separation constant used in
\cite{STU98} is instead $\lambda_{\rm STU} = \blambda+s(1+\alpha)$. The separation constant used by \cite{CTdC} is related by
\begin{equation}
\label{labla}
  \lambda_{\rm STU} = \lambda_{\rm CTdC} + s (1-\alpha) - 2 (1+\alpha)^2 a m \omega + (1+\alpha)^2 a^2 \omega^2.
\end{equation}
}
This separation constant $\blambda$ is determined by an eigenvalue problem for $S$, which for a given $s, m, \alpha$ and $a\omega$ is \cite{ChambersMoss}
\begin{equation}
\label{angop}
\left( -\partial_u \Delta_\theta \partial_u + V_\theta  \right) S = \blambda S
\end{equation}
in our conventions, where
\begin{align}
\label{angpot}
 W_\theta & = (1+\alpha) \left[ a\omega (1-u^2) -  m \right], \\
 V_\theta & =  2(2s^2+1)\alpha u^2 + \Delta_\theta^{-1} (W_\theta+s\partial_u \Delta_\theta/2)^2-2 s \partial_u W_\theta,
\end{align}
using $u=\cos \theta$. The boundary condition for this eigenvalue problem is that $S(u)$
be finite in both limits $u \to \pm 1$. In the case $s= \pm 2$, this geometrically corresponds to the regularity of the perturbed Weyl tensor or Hertz potential in a suitable Newman-Penrose tetrad that is regular at the poles $\theta = 0,\pi$. For real $\omega$ 
the angular operator \eqref{angop} is self-adjoint, so $\blambda$ is real.

It is standard to label the angular separation constant $\blambda \equiv {}_s \blambda_{\ell m}(a\omega,\alpha)$ and eigenfunctions for a given $s, m, \alpha$ and $a\omega$ by a natural number $\ell$ greater than or equal to the maximum of $|s|,|m|$, such that, when $\alpha=a\omega=0$, we have ${}_s\blambda_{\ell m}(a\omega=0,\alpha=0) = \ell(\ell+1) -s^2$ in our specific conventions.

\section{Transformation theory for radial Teukolski equation}
\label{sec:transformationtheory}

\subsection{Heun form of radial equation}
\label{Heunradialsec}

In this paper, Heun's differential equation is written as
\begin{equation}
\label{eq:HeunIsomonodromy}
 \partial_z^2 y(z) + \left( \frac{1-\eta_0}{z} + \frac{1-\eta_1}{z-1} + \frac{1-\eta_x}{z-x} \right) \partial_z y(z)
 + \left[ \frac{k_1(k_2+1)}{z (z-1)} - \frac{x (x-1) K_x}{z (z-1) (z-x)} \right] y(z) = 0 , 
\end{equation}
where $k_1, k_2$ are parameters such that $k_1+k_2 = -(\eta_0+\eta_1+\eta_x-1)$. See \cite{Ronveaux} for 
many facts about Heun's equation. The relation between $(K_x,\eta_0,\eta_1,\eta_x,k_1,k_2)$ 
and $(q_x, \alpha, \beta, \gamma, \delta)$ as used in \cite{Ronveaux} is linear and immediate.

As noticed in \cite{STU98}, the radial Teukolsky equation can be transformed to Heun form by a suitable substitution including a fractional linear change of the coordinate $r$ to a coordinate $z$, which maps $r_c$ to $z=0$, $r_+$ to $z=x$, $r_-$ to $z=1$, and $r_o$ to $z = \infty$, by
\begin{align}
\label{zr}
 z(r) & = z_\infty \frac{r-r_c}{r-r_o}, &
 z_\infty & = \frac{r_- - r_o}{r_- - r_c}, &
 x & = z_\infty \frac{r_+ - r_c}{r_+ - r_o}.
\end{align}
We furthermore define (our $\eta$'s deviate by a factor $2$ from the $\theta$'s of \cite{NMLC})
\begin{subequations}
\label{eq:etadef}
\begin{align}
 \eta_0 & = - i \frac{\omega - m \omega_c}{\kappa_c} + s, &
 \eta_x & = i \frac{\omega - m \omega_+}{\kappa_+} + s, \\
 \eta_1 & = - i \frac{\omega - m \omega_-}{\kappa_-} + s, &
 \eta_\infty & = i \frac{\omega - m \omega_o}{\kappa_o} + s.
\end{align}
\end{subequations}
After the substitution \cite{STU98}
\ben
 y(z) = z^{\frac{\eta_0}{2}+\frac{s}{2}} (z-1)^{\frac{\eta_1}{2}+\frac{s}{2}} (z-x)^{\frac{\eta_x}{2}+\frac{s}{2}} (z - z_\infty)^{-(2s+1)} R[r(z)],
\een
the radial equation turns into the Heun equation \cite{STU98}, in the form \eqref{eq:HeunIsomonodromy}
where $k_1, k_2$ are given by
\begin{align}
 k_1 & = 1 - 2 s, &
 k_2 & = \eta_\infty - 2 s.
\end{align}
$K_x$ is also conventionally called the ``accessory parameter''. For the radial equation, it is determined by $m$, $\omega$, $s$, the separation constant $\blambda$, and the parameters of the spacetime according to  
\begin{multline}
 K_x = \frac{(1-2s)(\eta_0 + \eta_x - 2 s)+4s}{2x} + \frac{(1-2s)(\eta_1 + \eta_x - 2 s)+4s}{2(x-1)} - \frac{(1 + s)(1 + 2s)}{x-z_\infty} \\
 + \frac{L^2 [ \blambda -s(1-\alpha) ] + 2 (1+s)(1+2s) r_+^2}{x(x-1)(r_c - r_-)(r_+ - r_o)} .
\end{multline}
We note the relation
\ben
\label{eq:sum_theta}
 \eta_0 + \eta_1 + \eta_x + \eta_\infty = 4 s,
\een
which is implicitly used in various manipulations below concerning the radial equation. 

\subsection{Boundary conditions for radial equation}
The parameters $\eta_i$, $i \in \{ 0, x, 1, \infty \}$, are the local monodromy coefficients for single loops around the corresponding pole, corresponding the difference of the Frobenius exponents. To be more precise, the behavior near $z_i \in \{ 0, x, 1 \}$ of solutions to \eqref{eq:HeunIsomonodromy} is a linear combination of solutions of the form
\ben
\label{eq:y_Asymptotics}
 (z- z_i)^{\eta_i} (1 + \cO(z-z_i)), \qquad 1 + \cO(z-z_i),
\een
whereas at $\infty$, the local behaviour is of the form
\ben
 z^{-1+2s} (1 + \cO(z^{-1})), \qquad z^{- 1 +2s - \eta_\infty} (1 + \cO(z^{-1})). 
\een
Here $\cO(z)$ stands for a function vanishing at $0$ and analytic in a neighborhood thereof.
We note that for $\eta_i \in \Z$, the more singular solution (the first of \eqref{eq:y_Asymptotics} for $\eta_i < 0$ and the second one for $\eta_i > 0$) receives a logarithmic correction. Furthermore, the less singular solution is hidden in the subleading tail of the leading one. This situation occurs only for an exceptional set of parameters and is not a concern because 
our main arguments use a continuity argument to extrapolate from exceptional cases. 

According to the above asymptotic formulas \eqref{eq:y_Asymptotics}, the local behavior of $R(r)$ near $r_+$ and $r_c$ is by \eqref{zr} a linear combination of 
\begin{subequations}
\label{Rbcs}
\begin{align}
 R^{(c)}_{\rout / \rin} & = (r_c - r)^{\pm \frac{\eta_0}{2}-\frac{s}{2}} (1 + \cO(r-r_c)) \propto e^{\pm i (\omega - m \omega_c ) r_*} e^{\kappa_c (\mp s + s) r_*} (1 + \cO(r-r_c)), \\
 R^{(+)}_{\rup / \rdown} & = (r - r_+)^{\pm \frac{\eta_x}{2}-\frac{s}{2}} (1 + \cO(r-r_+)) \propto e^{\pm i (\omega - m \omega_+ ) r_*} e^{\kappa_+ (\pm s - s) r_*}  (1 + \cO(r-r_+)),
\end{align}
\end{subequations}
using the conventional terminology ``in/out/up/down'' in black hole perturbation theory, where the tortoise coordinate $r_*$ is defined in the case of Kerr-dS by
\ben
 r_* = \frac{1}{2 \kappa_+} \log |r - r_+| - \frac{1}{2 \kappa_-} \log |r - r_-| - \frac{1}{2 \kappa_c} \log |r_c - r| 
\een
between $r_+$ and $r_c$.

Writing a general solution near $r = r_+$ and $r = r_c$ in terms of the above solutions with invertible ``connection matrices'' $C^{(+)}$, $C^{(c)}$, i.e.,
\begin{equation}
\label{eq:Phi_c_+}
    \Phi = \begin{pmatrix} R^{(c)}_\rout & R^{(c)}_\rin \end{pmatrix} C^{(c)} = \begin{pmatrix} R^{(+)}_\rup & R^{(+)}_\rdown \end{pmatrix} C^{(+)} ,
\end{equation}
we see that a quasi-normal mode (QNM, outgoing at the cosmological and downgoing at the event horizon) must fulfill
\begin{align}
 R_{QNM} & \propto \Phi (C^{(c)})^{-1} \begin{pmatrix} 1 \\ 0 \end{pmatrix}, &
 R_{QNM} & \propto \Phi (C^{(+)})^{-1} \begin{pmatrix} 0 \\ 1 \end{pmatrix},
\end{align}
which is consistent if and only if the two equivalent conditions
\begin{align}
\label{eq:QNM_c_+}
    C^{(c)} (C^{(+)})^{-1} & = \begin{pmatrix} * & * \\ * & 0 \end{pmatrix}, &
    C^{(+)} (C^{(c)})^{-1} & = \begin{pmatrix} 0 & * \\ * & * \end{pmatrix},
\end{align}
are fulfilled. For anti-quasinormal modes (AQNM, ingoing at the cosmological and upgoing at the event horizon), the right hand sides of these equations are interchanged. For further detail see \cite{Ronveaux}, section A.3.5, where QNMs are Heun functions of so-called type II with respect to the singularities $0$, $x$, while AQNMs are of so-called type III.

\subsection{Setup of isomonodromy method }
\label{sec:Setup}
Mass symmetry transformations of the spectrum of the Teukolsky equation on Kerr-dS were first suggested by \cite{Hatsuda:2020sbn}. 
In \cite{CTdC} they were established 
by a different method based on the Leaver \cite{Leaver:1985ax} and MST \cite{STU98} approaches to spectral problems for the Teukolsky equation.\footnote{See also \footnote{Umetsuissue}.}
Here we reconsider them using the so-called isomonodromy method, which we briefly describe first. 

The first step in the isomonodromy approach (see \cite{Jimbo:1981tov}, \cite{JimboMiwaII}, \cite{NMLC} as general references, which contain references to the historical background) is to consider a certain deformation
\begin{multline}
\label{eq:FuchsIsomonodromy}
  \partial_z^2 y(z) + \left( \frac{1-\theta_0}{z} + \frac{1-\theta_1}{z-1} + \frac{1-\theta_t}{z-t} - \frac{1}{z - q} \right) \partial_z y(z) \\
  + \left[ \frac{k_1 (k_2 + 1)}{z (z-1)} - \frac{t (t-1) H_x}{z (z-1) (z-t)} + \frac{q (q-1) p}{z (z-1) (z - q)} \right] y(z) = 0.
\end{multline}
of the Heun equation \eqref{eq:HeunIsomonodromy}, where $p=p(t), q=q(t)$, 
and $H_x=H_x(p(t),q(t),t)$ have a special dependence on the interpolation parameter, $t$, specified below. (Note that $t$ must not be confused with the BL coordinate $t$ in the Kerr-dS line element!) The constants $k_1, k_2$ are to be related to the $\theta$'s by
\ben
\label{eq:kappa_theta}
 k_1 + k_2 = - (\theta_0 + \theta_1 + \theta_t)
\een
and the Frobenius exponents are $0$ and $\theta_i$ for $i=0,1,t$ and $k_1$ and $k_2 + 1$ near $\infty$. $H_x$ is determined by demanding that solutions have no branch cuts at $z = q$---so that $z = q$ is only an ``apparent'' singularity in the standard terminology---which, as it turns out, is the case iff
\ben
\label{eq:K_constraint}
 H_x(p,q,t) = \frac{q (q - 1) (q - t)}{t (t-1)} \left[p^2 - \left( \frac{\theta_0}{q} + \frac{\theta_1}{q-1} + \frac{ \theta_t-1}{q-t} \right) p + \frac{k_1 (k_2 + 1)}{q (q-1)} \right].
\een 
The interpolated equation \eqref{eq:FuchsIsomonodromy} turns into the original Heun equation \eqref{eq:HeunIsomonodromy} in different limits provided that appropriate conditions are imposed onto $p, q$, and furthermore, if we impose that the interpolated equation have the same monodromies around the poles $z=0,t,1,\infty$ as the original equation. First of all, we should set $t = x$ in order to match the position of the non-apparent singularities. In order to get rid of the supplementary term $1/(z-q)$ in the coefficient of $y'$, we can either choose $q(t=x)$ to coincide with either $0, 1, t$ or consider $q \to \infty$. It turns out that the choice $q(t=x) = t = x$ is most useful. 
In order to match \eqref{eq:FuchsIsomonodromy} with \eqref{eq:HeunIsomonodromy} at $t=x$, we then have to choose \cite{CarneirodaCunha:2015qln}, \cite{daCunha:2015ana}
\begin{align}
\label{eq:theta_vartheta}
 \theta_0 & = \eta_0, &
 \theta_1 & = \eta_1, &
 \theta_t & = \eta_x - 1, &
  p & = - \frac{K_x}{\theta_t}.
\end{align}
The relation \eqref{eq:kappa_theta} is then satisfied both for the radial and angular equation for the respective choices of the $\eta$'s and $k$'s.

Next, we consider the monodromies of the interpolated equation \eqref{eq:FuchsIsomonodromy}, which one takes to remain unchanged under the flow in $t$ in the isomonodromy approach. This discussion is by far most transparent if one appropriately rewrites the single, second order differential equation \eqref{eq:FuchsIsomonodromy} as a component of a first order $2 \times 2$ matrix equation with rational coefficients having only first order poles at $z=0,1,t,\infty$:
\ben
\label{eq:FuchsEqn}
 \del_z Y(z) = A(z) Y(z),
\een
with
\begin{align}
 Y(z) & = \begin{pmatrix} y(z) \\ w(z)\end{pmatrix}, \\
\label{eq:A}
 A(z) & = \frac{A_0}{z} + \frac{A_1}{z-1} + \frac{A_t}{z-t},
\end{align}
where
\ben
\label{eq:A_infty}
 A_\infty = - (A_0 + A_1 + A_t) = \diag(k_1, k_2)
\een
is a diagonal matrix. The matrices $A_i$ can be chosen such that $A_i$ has eigenvalues $\{ 0, \theta_i \}$ for $i \in \{ 0, 1, t \}$. Then the Frobenius exponents are $0$ and $\theta_i$ for $i = 0, 1, t$ and $k_1$ and $k_2$ near $\infty$. The $A_i$ depend on the coefficients of the 
interpolating equation \eqref{eq:FuchsIsomonodromy}, $\theta_i, k_i, K_x, p, q$, and their explicit form is given in\footnote{To compare with their formulas, one has to put $q \to y, p \to z, k_i \to \kappa_i$.} \cite{JimboMiwaII}, equation C.47. However, without having to consider  the explicit form of the matrices $A_i$, one can note that, by taking the trace of \eqref{eq:A_infty}, we have \eqref{eq:kappa_theta}. Following \cite{JimboMiwaII}, we introduce
\ben
\label{eq:kappa_vth_infty}
 \theta_\infty = k_1 - k_2.
\een
Hence, with the above assignments, we have
\ben
\label{eq:theta_varthetai}
 \theta_\infty = - \eta_\infty + 1.
\een
Some of the results from the literature that we are going to use are available for a slightly modified system with coefficient matrices $\tilde A_i$ having eigenvalues $\pm \theta_i/2$. We may easily transform to such a system by considering
\ben
 \tilde Y(z) = z^{-\theta_0/2} (z-1)^{-\theta_1/2} (z-t)^{-\theta_t/2} Y(z)
\een
and $\tilde A_i = A_i - \frac{\theta_i}{2} 1_2$ for $i \in \{0, 1, t \}$, as well as $\tilde A_\infty = \diag( \frac{\theta_\infty}{2}, - \frac{\theta_\infty}{2})$.

The isomonodromic flow is obtained by imposing a further differential equation on the system in the flow parameter $t$, i.e. considering $\tilde Y$
as a function of both $(z,t)$ and postulating the overdetermined system
\ben
\label{Schl1}
\begin{split}
    \del_t \tilde Y(z,t) &= \tilde B(z,t)  \, \tilde Y(z,t),\\
    \del_z \tilde Y(z,t) &= \tilde A(z,t)  \, \tilde Y(z,t),\\
\end{split}
\een
where 
\ben
     \tilde B(z,t) = - \frac{\tilde A_t(t)}{z-t} .
\een
Imposing consistency onto the overdetermined system \eqref{Schl1} is well-known to be equivalent to the requirement 
that the $\tilde A_i$ solve the ``Schlesinger equations''
\begin{align}
\label{Schl2}
    \del_t \tilde A_0 & = \frac{1}{t} [\tilde A_t, \tilde A_0], &
    \del_t \tilde A_1 & = \frac{1}{t-1} [\tilde A_t, \tilde A_1], &
    \del_t \tilde A_t & = \frac{1}{t} [\tilde A_0, \tilde A_t] + \frac{1}{t-1} [\tilde A_1, \tilde A_t].
\end{align}
These not only guarantee that the above conditions on the eigenvalues of the $\tilde A_i$  are preserved under the flow of $t$, but also that the monodromies of the system are preserved. By definition, the monodromies are defined as the matrices $\tilde M[\gamma]$ such that a fundamental solution $\tilde Y(z,t)$ (a row of two column vectors of two linearly independent solutions to \eqref{eq:FuchsEqn}, i.e., a $2 \times 2$ matrix), when analytically continued by the usual method of overlapping neighborhoods along a closed loop $\gamma$ avoiding any pole, is right-multiplied by $\tilde M[\gamma]$ having gone counterclockwise around the loop once. By the usual arguments, the monodromy matrices $\tilde M[\gamma]$ only depend on the homotopy class of the loop on the punctured Riemann sphere $\mathbb{P}^1 \setminus \{0,t,1,\infty\}$. Since the fundamental solution is only unique up to a constant similarity transformation conjugation, the monodromy is properly speaking a homomorphism from $\pi_1(\mathbb{P}^1 \setminus \{0,t,1,\infty\}) \to SL_2(\mathbb{C})/\sim$. 

From a physics viewpoint, the isomonodromic 
deformation property guaranteed by equations \eqref{Schl1} respectively \eqref{Schl2}, is best understood as a kind of zero curvature condition for an appropriate $SL_2(\mathbb{C})$ connection, $\boldsymbol{A}=\tilde A \ud z + \tilde B \ud t$. From this perspective, the monodromy matrices are the holonomies of closed loops in $(\mathbb{P}^1 \setminus \{0,t,1,\infty\})_z \times \mathbb{R}_t$, i.e., 
\begin{equation}
\tilde M[\gamma] = {\rm P} \exp \left( \oint_\gamma \boldsymbol{A} \right)
\end{equation}
for a loop $\gamma$ having a constant $t$.
For a flat connection, the traces of these are clearly homotopy invariant, and thus in particular invariant under a variation of $t$.

Using the parameterization of the matrices $A_i$ of \cite{JimboMiwaII}, equation C.47, equations \eqref{Schl2} can be seen to be equivalent to the Hamiltonian system \eqref{eq:Hamilton} for $p,q$ given below, and these in turn are equivalent to a Painlev\'e VI equation, see \cite{JimboMiwaII}, equation C.56.

Following \cite{JimboMiwaII, Jimbo}, we fix a fundamental solution $\tilde Y$  normalized by the behavior at infinity as
\ben
\label{eq:AsymptoticCondition}
 \tilde Y = (1 + \cO(z^{-1})) \diag(z^{-\frac{\theta_\infty}{2}}, z^{\frac{\theta_\infty}{2}}).
\een
Next, one defines the so-called ``connection matrices'' $C^{(i)}$ by the local behaviour of this fundamental solution near the singularities at $z = z_i \in \{ 0, t, 1 \}$:
\ben
\label{eq:DefConnectionMatrices}
 \tilde Y(z, t) = (G^{(i)}(t) + \cO(z- z_i)) \diag( (z-z_i)^{\frac{\theta_i}{2}}, (z-z_i)^{-\frac{\theta_i}{2}} ) C^{(i)} \qquad z \to z_i.
\een
The matrices $G^{(i)}$ and $C^{(i)}$ are invertible, and the connection matrices $C^{(i)}$ are independent of $t$. The monodromy matrices corresponding to loops around the poles $z_i$ can be expressed in terms of the connection matrices by 
\ben
\label{eq:Monodromy}
 \tilde M^{(i)} = ( C^{(i)} )^{-1} \diag( e^{i \pi \theta_i}, e^{- i \pi \theta_i} ) C^{(i)}.
\een
By construction, they are independent of $t$, as are the composite monodromies. Note that for a fixed choice of $\tilde Y$, the definition \eqref{eq:DefConnectionMatrices} only fixes the connection matrices up to multiplication with an invertible diagonal matrix from the left. 

In the case when either a QNM respectively an AQNM mode exists between two singular points $z=0,x$, or equivalently between $r=r_+,r_c$ in terms of the radial BL coordinate $r$, we had to have
\ben
 C^{(+)} (C^{(c)})^{-1} = \begin{pmatrix} * & * \\ * & 0 \end{pmatrix} \quad \text{respectively} \quad \begin{pmatrix} 0 & * \\ * & * \end{pmatrix},
\een
Using these relations, one easily verifies that
\ben
\label{eq:QNM_necessaryCondition}
 \Tr ( \tilde M^{(+)} \tilde M^{(c)} ) = 2 \cos ( \pi (\theta_0 - \theta_t )).
\een
In particular, defining $\sigma_{i j}$ via
\ben
\label{eq:Def_sigma}
 2 \cos( \pi \sigma_{ij}) = \Tr ( \tilde M^{(i)} \tilde M^{(j)} ),
\een
we obtain that \cite{NMLC}
\ben
\label{eq:thetaCondition}
 \theta_0 - \theta_t \pm \sigma_{0t} \in 2 \Z
\een
is a necessary condition for the occurrence of either a QNM or AQNM between the singular points $r_+$ and $r_c$. Conversely, if \eqref{eq:thetaCondition} and $\theta_0, \theta_t \not \in \Z$ (which we assume for the time being) holds, then there must be either a QNM or an AQNM mode between $r_+$ and $r_c$).

Inserting the values of the $\theta$'s for the radial equation, we see that a necessary condition for this to happen is that \cite{NMLC}
\begin{equation}
\label{eq:QNM_theta_sigma}
    \theta_0 - \theta_t \pm \sigma_{0t} = \eta_0 - \eta_x + 1 \pm \sigma_{0t} = - i \left( \frac{1}{\kappa_c} + \frac{1}{\kappa_+} \right) \omega + i m \left( \frac{\omega_c}{\kappa_c} + \frac{\omega_+}{\kappa_+} \right) + 1 \pm \sigma_{0t} \in 2 \Z.
\end{equation}
Below we will often use the notation $\sigma \defeq \sigma_{0t}$. As is obvious from the definition \eqref{eq:Def_sigma}, we can and always will choose $\sigma$ such that $0 \leq \Real \sigma < 1$. It is also possible to reverse this argument, showing that, whenever $\theta_0 - \theta_t \pm \sigma \in 2 \Z$, then we have a QNM or an AQNM, although through this equation we cannot a priori detect whether it is one or the other.

\subsection{Mass symmetry transformations}

We will now show, using the method of isomonodromic deformations, 
that the composite monodromy parameter $\sigma = \sigma_{0t}$ 
in the Heun equation has certain symmetries. First of all, 
since we have $k_1-k_2=-\theta_\infty,k_1+k_2=-(\theta_1+\theta_2+\theta_t)$, and since the Heun equation 
is written in terms of $k_1,k_2$, the $\theta_i$, and $K_x$, 
we can view the coefficients in the Heun equation and $\sigma$
as a function $\sigma=\sigma(\btheta, K_x)$ of the parameters
$(\btheta, K_x) \equiv (\theta_0, \theta_t, \theta_x, \theta_\infty , K_x)$.

To connect these parameters with $\sigma$, an important role is played by the so-called $\Sigma_{\rm VI}$-function, defined as \cite{JimboMiwaII}
\begin{equation}
 \Sigma_{\rm VI} 
 = t(t-1) H_x + t(m_3m_4-m_2m_4-m_2m_3) -\frac{1}{2}\big[m_3m_4+m_1m_2-(m_2+m_1)(m_3+m_4)\big],
\end{equation}
where $H_x$ is the Hamiltonian \eqref{eq:K_constraint}, and 
where we have defined\footnote{Our $m_i$ are related to $\nu_i$ of \cite{JimboMiwaII}, (C.60) by $m_1 = \nu_4$, $m_2 = - \nu_3$, $m_3 = - \nu_1$, $m_4 = - \nu_2$.}
\begin{align}
\label{eq:Def_m_i}
 m_1 & = \frac{1}{2} (\theta_0 - \theta_1), &
 m_2 & = \frac{1}{2}(\theta_1 + \theta_0), &
 m_3 & = \frac{1}{2}(- \theta_\infty - \theta_t), &
 m_4 & = \frac{1}{2}(\theta_\infty - \theta_t).
\end{align}
As a consequence of the Hamilton equations
\begin{equation}
\label{eq:Hamilton}
    \dot p = - \frac{\partial H_x(p,q,t)}{\partial q}, 
\quad 
\dot q =  \frac{\partial H_x(p,q,t)}{\partial p}, 
\end{equation}
for the functions $p(t), q(t)$ expressing the isomonodromic property of the flow in $t$, the function $\Sigma_{\rm VI}$ solves a second order differential equation in $t$ equivalent to the Painlevé VI equation. This equation can be expressed solely in terms of the $m_i$, and is invariant under any permutation of the $m_i$, as well as the simultaneous sign change of any even number of them, see \cite{JimboMiwaII}, (C.61):
\begin{equation}
\label{sigma6}
\begin{split}
   & \frac{\ud}{\ud t} \Sigma_{\rm VI} \left[t(t-1) \frac{\ud^2}{\ud t^2} \Sigma_{\rm VI} \right]^2+
   \left[
2 \frac{\ud}{\ud t} \Sigma_{\rm VI} \left(t \frac{\ud}{\ud t} \Sigma_{\rm VI} -
\Sigma_{\rm VI} \right) - \left( \frac{\ud}{\ud t} \Sigma_{\rm VI}\right)^2
- m_1m_2m_3m_4
   \right]^2\\
   &=\prod_{i=1}^4 \left(\frac{\ud}{\ud t} \Sigma_{\rm VI}+m_i^2 \right) .
\end{split}   
\end{equation}
In the following, we will exploit this symmetry to 
infer a certain symmetry of $\sigma=\sigma(\btheta, K_x)$.

To establish a relation between $\sigma \equiv \sigma_{0t}$ and $\Sigma_{\rm VI}$, we first use the asymptotic behavior as $t \to 0$, which can be extracted from (2.15) of \cite{Jimbo},
\begin{subequations}
\label{eq:Jimboasy}
\begin{align}
    q & = - 
    \frac{(\theta_\infty + \theta_1 +  \sigma)^2 [ \theta_0^2 - (\theta_t +  \sigma)^2]}{4  \sigma^2 [(\theta_\infty + \theta_1)^2 -  \sigma^2]} \tilde\varsigma^{-1} t^{1- \sigma} + \cO(t), \\
    q p & = \frac{\theta_0 + \theta_t -  \sigma}{2} + \cO(t^{\sigma})
\end{align}
\end{subequations}
where $\tilde \varsigma = \tilde\varsigma(\sigma=\sigma_{0t}, \sigma_{1t}, \btheta)$
is given by
\begin{equation}
\label{eq:hat_s}
    \tilde\varsigma = \frac{
    \Gamma^2(1-\sigma) \prod_\pm
    \Gamma(\frac{1}{2} (\pm \theta_0 + \theta_t + \sigma) +1) 
    \Gamma(\frac{1}{2} (\pm \theta_\infty + \theta_1 + \sigma) +1) 
    }{
    \Gamma^2(1+\sigma) \prod_\pm 
    \Gamma(\frac{1}{2} (\pm \theta_0 + \theta_t - \sigma) +1) 
    \Gamma(\frac{1}{2} (\pm \theta_\infty + \theta_1 - \sigma) +1) 
    } \varsigma
\end{equation}
and where $\varsigma$ is the 
function of $(\sigma=\sigma_{0t},\sigma_{1t}, \btheta)$ solving  the quadratic equation
\begin{multline}
\label{eq:s1}
    \sum_\pm {\varsigma}^{\pm 1} \prod_{\pm'}\sin \frac{\pi}{2} (\theta_0 \pm' \theta_t \mp \sigma) 
    \sin \frac{\pi}{2} (\theta_\infty \pm' \theta_1 \mp \sigma) 
    \\
    = 2 \prod_{\pm} \sin \frac{\pi}{2} (\theta_0 \pm \theta_t) \sin \frac{\pi}{2} (\theta_\infty \pm \theta_1) - \frac{1}{2} \sin^2 \pi \sigma \cos \pi \sigma_{1t}\\
    -\frac{1}{2}(\cos \pi \sigma -1)
\left( \cos \pi \theta_t \cos \pi \theta_\infty  +\cos \pi \theta_0 \cos \pi \theta_1 \right) .
\end{multline} 
The asymptotic formulas \eqref{eq:Jimboasy} for $p,q$ are valid under the assumptions 
\begin{align}
\label{eq:Constraints}
    \theta_0, \theta_t, \theta_1, \theta_\infty & \not\in \Z, &
    0 & \leq \Real \sigma < 1, &
    \frac{1}{2} (\theta_\infty \pm \theta_1 \pm' \sigma) & \not\in \Z, &
    \frac{1}{2} (\theta_0 \pm \theta_t \pm' \sigma) & \not\in \Z, 
\end{align}
and $\sigma \neq 0$. Similar formulas 
hold for $\sigma = 0$. They involve 
$\log t$ and the expansion coefficients of $\tilde \varsigma = 1 + \tilde \varsigma_1 \sigma + \cO(\sigma^2)
$, and can be extracted from \cite{Jimbo}.

Substituting the asymptotic formulas into the definition of 
$\Sigma_{\rm VI}$, one obtains \cite{Jimbo}, for $t \to 0$,
\ben
\label{eq:Sigmalim}
\Sigma_{\rm VI}(t) =  - \frac{1}{4}\sigma^2 + \frac{1}{4} ( m_1^2 + m_2^2 + m_3^2 + m_4^2 ) + 
\begin{cases} 
\cO\left(|t|^{1-\Real(\sigma)} \right) & \text{for $\sigma \neq 0$,}\\
\cO\left(|t \log t|^2 \right) & \text{for $\sigma = 0$.}
\end{cases}
\een
 Using the initial values for the dynamical variables $p(t), q(t)$ at $t=x$, 
we can also obtain the initial conditions for $\Sigma_{\rm VI}(t)$ and its first $t$-derivative at $t=x$. These turn out to be:
\ben
\label{eq:sigma_VI_x}
 \Sigma_{\rm VI}|_{t = x} = x(x-1) K_x + x 
 (m_3m_4-m_2m_3-m_2m_4) + \frac{1}{2}\big[-m_1m_2-m_3m_4 + (m_2+m_1)(m_3+m_4)\big]
\een
and, using the Hamilton equations \eqref{eq:Hamilton} and the relations between the $\theta$'s and the $k$'s, 
\ben
\label{eq:dsigma_VI_x}
 \frac{\ud}{\ud t} \Sigma_{\rm VI} |_{t = x} = -m_4^2.
\een
Consider now the following discrete group $D_3$ acting on $m_1, m_2, m_3, m_4$
as
\begin{equation}
\label{eq:D3}
    D_3: \quad m_i \mapsto \hat m_i := \pm_i m_{\pi(i)}, \quad i=1,2,3,4, 
    \quad \prod_{i=1}^4 (\pm_i 1) = 1, \quad \pi(4)=4,
\end{equation}
where $\pi$ is a permutation on the three elements $\{1,2,3\}$ but not $4$, and where the signs $\pm_i$ can be chosen independently of $i=1,2,3,4$ but there are must be an even number of minus signs. By what we have already said, the differential equation 
for $\Sigma_{\rm VI}$ is invariant under \eqref{eq:D3}. The initial condition \eqref{eq:dsigma_VI_x} also is invariant under \eqref{eq:D3}. 
The initial condition \eqref{eq:sigma_VI_x} is invariant 
under \eqref{eq:D3} provided that we change the accessory parameter $K_x$
according to 
\begin{equation}
\label{eq:Kchange}
\begin{split}
    K_x \mapsto \hat K_x := 
    K_x 
    &+  
 \frac{m_3m_4-m_2m_3-m_2m_4}{x-1} + \frac{-m_1m_2-m_3m_4 + (m_2+m_1)(m_3+m_4)}{2x(x-1)} \\
 &-
 \frac{\hat m_3\hat m_4-\hat m_2\hat m_3-\hat m_2\hat m_4}{x-1} - \frac{-\hat m_1\hat m_2-\hat m_3\hat m_4 + (\hat m_2+\hat m_1)(\hat m_3+\hat m_4)}{2x(x-1)}.
 \end{split}
\end{equation}
By uniqueness of solutions to second order ordinary differential equations, it follows that $\Sigma_{\rm VI}(t,\bm, K_x)$ is equal to $\Sigma_{\rm VI}(t,\hat \bm, \hat K_x)$ for any 
operation in $D_3$, where $\bm = (m_1,m_2,m_3,m_4)$
and similarly for $\hat \bm$. This symmetry hence also persists in the limit $t \to 0$, which exists if the constraints \eqref{eq:Constraints} hold. 
Furthermore, since in the limit \eqref{eq:Sigmalim}, the 
combination $m_1^2+m_2^2+m_3^2+m_4^2$ clearly is invariant under $D_3$, 
it follows that also $\sigma^2=\sigma(\bm,K_x)^2$ is invariant. 

Finally, $\sigma=\sigma(\bm,K_x)$ is clearly continuous as a function 
of its parameters, because these are just the parameters of the Heun equation, and the traces of the monodromies change continuously with those parameters. It follows that $\sigma^2=\sigma(\bm,K_x)^2$ is invariant for {\it all} values of the parameters. We therefore have obtained the following proposition:

\begin{proposition}
\label{propmasssym}
    The composite monodromy parameter $\sigma$, viewed as a function of the parameters $\bm, K_x$ in the Heun equation via equation \eqref{eq:Def_m_i}, satisfies 
    \begin{equation}
        \sigma(m_1,m_2,m_3,m_4,K_x)
= \pm \sigma(\hat m_1,\hat m_2,\hat m_3,\hat m_4,\hat K_x)
    \end{equation}
    for any $D_3$ transformation \eqref{eq:D3} and a corresponding change
    \eqref{eq:Kchange} of the accessory parameter.
\end{proposition}
\noindent
{\bf Remark.} We mention that symmetries of $\Sigma_{\rm VI}$ under various exchanges of its  arguments $\btheta$ are well-known in the literature \cite{Okamoto,Jimbo}, see e.g., \cite{Gamayun:2012ma} for a summary. However, when applying these, one has to be careful what exactly one considers as the arguments of the $\Sigma_{\rm VI}$-function in addition to $t$. For example, in \cite{Gamayun:2012ma}, $\Sigma_{\rm VI}$
is considered as a function of $(\btheta, \varsigma, \sigma)$ (via the isomonodromic tau function), while we consider the arguments besides $t$ to be $(\btheta, K_x)$. Of course,  \eqref{eq:s1} and the initial conditions \eqref{eq:sigma_VI_x}, \eqref{eq:dsigma_VI_x} of $\Sigma_{\rm VI}$ at $t=x$ provide relations between these quantities, but these relations are not fully symmetric in the $m$'s.

\medskip
\noindent

\subsection{Transformation of radial equation}
\label{sec:Rtransf}
We now connect proposition \ref{propmasssym} to the condition \eqref{eq:QNM_theta_sigma} for (A)QNMs of the radial Teukolsky equation, which as we recall can be restated as: ($\theta_0 - \theta_t \pm \sigma \in 2\Z$, with either sign $\pm$) $\Leftrightarrow$ (the Heun equation with parameters $\btheta, K_x$ has a QNM or AQNM). On the other hand it immediately follows from the definition \eqref{eq:Def_m_i} of the parameters $m_i$ that 
$\theta_0-\theta_t = m_1+m_2+m_3+m_4$, hence we have 
\begin{equation}
\begin{split}
    &\ \ \left(\pm \sigma +\sum_{i=1}^4 m_i  \in 2\Z \quad 
    \text{for some sign $\pm$} \right) \\
    \Longleftrightarrow & \ \ 
    (\text{the radial Heun equation with parameters} \  (\bm, K_x) \ 
    \text{has a QNM or a AQNM}).
\end{split}
\end{equation}
From the symmetric structure of this condition we immediately get the following corollary to proposition \ref{propmasssym}:

\begin{corollary}
Let $S_3 \times \Z_2$ be the subgroup of $D_3$ acting as 
$(m_1,m_2,m_3,m_4) \mapsto \pm(m_{\pi(1)},m_{\pi(2)},m_{\pi(3)},m_4)=(\hat m_1,\hat m_2,\hat m_3,\hat m_4)$, with a corresponding change of $\btheta, K_x$
as in equations \eqref{eq:Def_m_i}, \eqref{eq:Kchange},
where $\pi$ is any permutation on $\{1,2,3\}$. Then
\begin{equation}
\label{QNMinference}
\begin{split}
    &\ \ (\text{the radial equation with parameters} \  \btheta, K_x \ 
    \text{has a QNM or a AQNM}). \\
    \Longleftrightarrow & \ \ 
    (\text{the radial equation with parameters} \  \hat \btheta, \hat K_x \ 
    \text{has a QNM or a AQNM}).
\end{split}
\end{equation}    
\end{corollary}

We now show that for the Kerr-dS black hole, the potential  in a Schr\" odinger form of the radial Heun equation becomes real valued after the mass symmetry $m_2 \leftrightarrow m_3$ when $\omega$ is real. 

In terms of the $\theta$'s, this 
mass symmetry corresponds to the shift
$\hat \theta_j = \theta_j - \delta$, $j \in \{0,t,1,\infty\}$, where $\delta = \frac{1}{2} \sum_j \theta_j$.
Substituting the definitions $\theta_0 = \eta_0, \theta_1 = \eta_1, \theta_t = \eta_x-1, \theta_\infty = -\eta_\infty +1$, we obtain 
formulas for $\hat \eta_0, \hat \eta_1, \hat \eta_x$, and we see using 
\eqref{eq:etadef} that in those formulas, all the terms involving $s$ drop out\footnote{
We do not require $\hat \eta_\infty$, which could be determined as the solution to
\begin{equation}
    \hat k_1 ( \hat k_2 + 1) = ( 1 - 2 s) ( 1 - 2 s + \hat \eta_\infty ).
\end{equation}
}:
\begin{subequations}
\label{MassSym}
\begin{align}
 \hat \eta_0 & = \frac{1}{2} ( \eta_0 - \eta_1 + \eta_\infty - \eta_x ) =i\left( \frac{\omega-m\omega_-}{\kappa_-} - \frac{\omega-m\omega_+}{\kappa_+} \right) \\
 \hat \eta_x & =  \frac{1}{2} ( -\eta_1 - \eta_0 + \eta_\infty + \eta_x )
 =i\left( \frac{\omega-m\omega_c}{\kappa_c} + \frac{\omega-m\omega_-}{\kappa_-} \right) \\
  \hat \eta_1 & = \frac{1}{2} \left( \eta_\infty - \eta_x + \eta_1 - \eta_0 \right) =i\left( \frac{\omega-m\omega_c}{\kappa_c} - \frac{\omega-m\omega_+}{\kappa_+} \right), 
\end{align}    
\end{subequations}
which are purely imaginary for real $\omega$. 
Then we determine $\hat k_1, \hat k_2$ by
\begin{subequations}
\begin{align}
    \hat k_1 & = \frac{1}{2} \left( \hat \theta_\infty - \hat \theta_0 - \hat \theta_1 - \hat \theta_t \right) = 1 - \eta_\infty , \\
    \hat k_2 & = \frac{1}{2} \left( - \hat \theta_\infty - \hat \theta_0 - \hat \theta_1 - \hat \theta_t \right) = 2 s - \eta_\infty.
\end{align}
\end{subequations}
as well as (see \eqref{eq:Kchange})
\begin{equation}
\label{eq:Kchange1}
\begin{split}
    \hat K_x =& \,
    K_x 
    +  
 \frac{2(m_3-m_2)m_4}{x-1} + \frac{(m_4-m_1)(m_2-m_3)}{x(x-1)} 
 \end{split}
\end{equation}
It is most instructive to analyze the consequences of the mass symmetry transformation 
in terms of a Schr\" odinger form of the Heun radial equation: Instead of $y(z)$ we can consider 
\ben
u(z) = z^{\frac{1}{2}(1 - \eta_0)} (z-1)^{\frac{1}{2}(1 - \eta_1)} (z-x)^{\frac{1}{2}(1 - \eta_x)} y(z). 
\een
This substitution puts the Heun equation \eqref{eq:HeunIsomonodromy} into the 
Schrödinger form 
\begin{equation}
   (-\partial_z^2+V_z) u = 0
\end{equation}
with
\begin{multline}
\label{eq:V}
V_z(z) = 
 -\frac{1 - \eta_0^2}{4z^2} 
 - \frac{1 - \eta_1^2}{4(z-1)^2} 
 - \frac{1 - \eta_x^2}{4(z-x)^2} 
 +  \frac{(\eta_0 - 1)(\eta_1 - 1)}{2z(z-1)} 
 +  \frac{(\eta_0 - 1)(\eta_x - 1)}{2z(z-x)} \\ 
 +  \frac{(\eta_1 - 1)(\eta_x - 1)}{2(z-1)(z-x)} 
 - \frac{ k_1 (k_2 + 1) }{z (z-1)} 
 + \frac{x (x-1) K_x}{z (z-1) (z-x)}.
\end{multline}
The condition for QNMs is 
\begin{align}
 u(z) & = z^{- \frac{1}{2}\eta_0 + \frac{1}{2}}(1+{\mathcal O}(z)), &
 u(z) & = (z-x)^{\frac{1}{2} \eta_x + \frac{1}{2}}(1+{\mathcal O}(z-x)),
\end{align}
while the one for AQNMs is
\begin{align}
 u(z) & = z^{\frac{1}{2} \eta_0 + \frac{1}{2}}(1+{\mathcal O}(z)), &
 u(z) & = (z-x)^{-\frac{1}{2} \eta_x + \frac{1}{2}}(1+{\mathcal O}(z-x)).
\end{align}
To map this Schr\" odinger problem in the interval $z \in (0,x)$ between the singular points where the boundary conditions are imposed, 
to one on the entire real line, we now perform yet another transformation to $\tilde R = Z^{- \frac{1}{2}} u$, with $Z = z (1-z) (x-z)$, seen as a function of $\tilde z$ defined by $\ud \tilde z = Z^{-1} \ud z$. Note that the factor $Z$ has been chosen to be positive for $0 < z < x$. With respect to this variable, $\tilde R$ fulfills another Schr\" odinger equation,
\begin{equation}
\label{eq:HeunSch}
(-\partial_{\tilde z}^2  + V_{\tilde z}) \tilde R = 0 
\end{equation}
with
\ben
\label{eq:tilV}
 V_{\tilde z} = Z^2 \left[ V_z - \frac{1}{2} \frac{\partial_z^2 Z}{Z} + \frac{1}{4} \frac{(\partial_z Z)^2}{Z^2} \right].
\een
The QNM boundary conditions for $\tilde R$ are
\begin{equation}
\label{QNMosc}
 \tilde R \sim \begin{cases} e^{-\frac{1}{2} x \eta_0 \tilde z} & \tilde z \to - \infty, \\ e^{-\frac{1}{2}x (x-1) \eta_x \tilde z} & \tilde z \to + \infty, \end{cases}
\end{equation}
whereas those for AQNMs are
\begin{equation}
\label{AQNMosc}
 \tilde R \sim \begin{cases} e^{\frac{1}{2} x \eta_0 \tilde z} & \tilde z \to - \infty, \\ e^{\frac{1}{2} x (x-1) \eta_x \tilde z} & \tilde z \to + \infty. \end{cases}
\end{equation}
When considering the Schrödinger equation \eqref{eq:HeunSch} after a mass symmetry transformation, one has to replace the $\eta$'s, $k$'s, and $K_x$ both in \eqref{eq:V} and in \eqref{QNMosc}, \eqref{AQNMosc} by their hatted counterparts. Substituting these 
into the potential $ V_{\tilde z}$, (see \eqref{eq:tilV}, \eqref{eq:V}) one obtains yet another potential. 
An automated computation \cite{nb} shows that this potential is\footnote{
In order to make a comparison with \cite{nb}, one should identify 
$\hat V_{\tilde z}=${\tt V1324t}, as well as $r_+=${\tt rp}, $r_-=${\tt rm}, 
$r_c=${\tt rc}, $r_o=${\tt ro}, $\alpha=${\tt al}, $\blambda=${\tt lamb}, 
$\eta_0=${\tt t0}, $\eta_x=${\tt tx}, $\eta_1=${\tt t1}, $\eta_\infty=${\tt ti}.
}
\begin{multline}
\label{eq:tilVnew}
\hat V_{\tilde z}(\tilde z)  =
\frac{1}{4} \bigg[ 
z(x-1) (\eta_0-s) 
 + (z-1)x (\eta_1-s) 
 + (z-x)(\eta_x-s)
\bigg]^2
 -z(z-x)(x-1)(\eta_0-s)(\eta_x-s)\\
  - z(z-1)(z-x) \bigg\{ (1-s^2) z 
  - \frac{L^2 [\blambda  -s(1-\alpha)] + s(s+1) (r_+ + r_-)^2  - (s+1)(r_o r_c + r_+ r_-)}{(r_c - r_-)(r_+ - r_o)}
  \bigg\}.
  \end{multline} 
For real $\omega$ the terms $\eta_j-s$ are purely imaginary and 
the separation constant $\blambda$ is real.  Therefore, also $\hat V_{\tilde z}(\tilde z)$ is real for real $\tilde z$ values. For the asymptotic values of this potential as $\tilde z \to \pm \infty$, we find, 
\begin{equation}
\label{Vasympt}
    \hat V_{\tilde z}(\tilde z) \to
    \begin{cases}
    {\displaystyle
        -\frac{x^2}{4} \left( \frac{\omega-m\omega_-}{\kappa_-} - \frac{\omega-m\omega_+}{\kappa_+} \right)^2} & \text{as $\tilde z \to -\infty$,}\\
        \displaystyle{-\frac{x(1-x)}{4} \left( \frac{\omega-m\omega_c}{\kappa_c} - \frac{\omega-m\omega_+}{\kappa_+} \right)^2} & \text{as $\tilde z \to +\infty$,}
    \end{cases}
\end{equation}
consistent with the asymptotically oscillatory behaviors 
\eqref{QNMosc}, \eqref{AQNMosc} for (A)QNMs, replacing $\eta_0 \to \hat \eta_0$
and $\eta_x \to \hat \eta_x$ under the mass symmetry \eqref{MassSym}. 
The implication \eqref{QNMinference} therefore gives:

\begin{proposition}
\label{prop:2}
    The Heun equation in Schr\" odinger form \eqref{eq:HeunSch} with potential $V_{\tilde z}$ given by equations \eqref{eq:V}, \eqref{eq:tilV}
    has a QNM or AQNM (i.e. a solution $\tilde R$ with one of the  asymptotically oscillatory behaviors 
\eqref{QNMosc}, \eqref{AQNMosc}) if and only if the same equation 
with the new potential $\hat V_{\tilde z}$ given by \eqref{eq:tilVnew} has a QNM or AQNM (i.e. a solution $\tilde R$ with one of the  asymptotically oscillatory behaviors 
\eqref{QNMosc}, \eqref{AQNMosc} with $\eta_i$ replaced by $\hat \eta_i$ as in equation \eqref{MassSym}.)
\end{proposition}

By the usual symmetries of the radial Teukolsky equation, if $\omega$ is a QNM frequency for $(m,s)$, then $\bar \omega$ is an AQNM frequency for $(m,-s)$
\cite{Teukolsky:1973ha}. Thus, for real frequencies, proposition \ref{prop:2} gives a complete equivalence of the QNM spectrum under the mass symmetry. In their analysis of mode stability of Kerr-dS, \cite{CTdC} have considered similar, but not quite identical, mass symmetries as we have. Using their mass symmetries, \cite{CTdC}  have recast the radial Teukolsky equation into the form of a Schrödinger with real potential having the same asymptotic behavior as ours \eqref{Vasympt}. Using these properties, \cite{CTdC}  have shown that a real QNM frequency must satisfy
\ben
\label{marcritabound}
\hat \omega_c < \omega/m < \hat \omega_+, 
\een
where $\hat \omega_j$ are given by \eqref{homegadef}. 
Thus, by proposition \ref{prop:2}, our method reproduces the restriction \eqref{marcritabound} by \cite{CTdC}.

For non-real frequencies, proposition \ref{prop:2} does not yield an equivalence of the QNM spectrum alone under the mass symmetry because a QNM frequency before 
the mass symmetry transformation could correspond to an AQNM afterwards. This is because our condition on the monodromy parameters $\sigma \equiv \sigma_{0t}$ 
\label{eq:thetaCondition} does not distinguish QNMs and AQNMs. In this regard, our method is inferior, e.g., to the method by \cite{Umetsu} via integral transformations, which gives an equivalence of the QNM spectrum under the mass symmetries.

\section{Geometry of mass symmetry}
\label{sec:geometry}
In this section we give a geometric interpretation of the mass symmetry transformation as transforming the Teukolsky master equation to  
a scalar wave equation on a ``dual" spacetime. In fact, for the case of Kerr-dS, this dual spacetime again turns out to be a stationary black hole spacetime with a pair of bifurcate Killing horizons, and the potential is {\it real, positive definite} in a neighborhood including the region exterior to the black hole region. However, before we come to Kerr-dS, 
we first explain the nature of the dual spacetime in the case of Kerr, and its connection with the Kerr mode stability problem \cite{Whiting}. 

\subsection{Whiting's metric and mode stability of Kerr}

The sub-extremal Kerr black hole spacetime corresponds to 
the limit as $\Lambda \to 0$ at fixed $M,a$ in \eqref{BLKdS}. In this limit, the root $r_c \to \infty$ and $\Delta_r \to \Delta$, where
\ben
\Delta = (r-r_+)(r-r_-) = r^2 -2Mr + a^2.
\een
Let us set 
\ben
\label{FGdef}
\begin{split}
\epsilon& = \frac{r_+-r_-}{r_++r_-}= \sqrt{1-a^2/M^2},\\
F &= \left[ \frac{(r-r_-)^2}{(r-r_+)^2}  - \frac{a^2(r-M)^2}{M^2\Delta} \right] \frac{\Delta}{\epsilon^2}
+ a^2 \cos^2 \theta,\\
G &= a\left(\cos \theta + \frac{r-M}{\epsilon M} \right),
\end{split}
\een
In the context of his transformation theory for the Teukolsky equations on Kerr, \cite{Whiting} has implicitly introduced a certain 
metric, $\hat g_{ab}$ which is 
defined for $r>r_+$ as 
\ben
\label{Whitingds2}
\ud\hat s^2 = \sqrt{\Delta} G \sin \theta \left( \frac{-2G\ud t \ud\phi + F\ud\phi^2}{G^2}
+ \frac{\ud r^2}{\Delta}  + \ud \theta^2   \right) . 
\een
$T^a = (\partial_t)^a$ clearly is a null Killing VF of $\hat g_{ab}$. For $r>r_+$ we have $F,G>0$, so 
\eqref{Whitingds2} defines a smooth Lorentzian metric on the manifold ${\mathbb R}_t \times (0,\pi)_\theta \times (r_+,\infty)_r \times S^1_\phi$.
A number of other curious and non-obvious geometric properties of \eqref{Whitingds2} are pointed out for the benefit of an interested reader 
in appendix \ref{app:A}, but they are not needed for the discussion of mode stability.

The connection between Whiting's metric \eqref{Whitingds2} and the Kerr mode stability problem is as follows \cite{Whiting}. 
Let $\Psi = R(r)S(\theta)e^{-i\omega t+im\phi}$ be a putative QNM solution to the spin-$s$ Teukolsky master equation in Kerr in BL coordinates such that 
$\Imag(\omega)>0$, i.e., mode stability would not hold. By applying suitable integral transformations to $R,S$, \cite{Whiting} obtains 
new $\hat R, \hat S$, such that $\hat \Psi = \hat R(r)\hat S(\theta)e^{-i\omega t+im\phi}$ satisfies an equation which can be easily reformulated as a 
Klein-Gordon equation
\ben
\label{hKG}
(-\hat g^{ab} \hat \nabla_a \hat \nabla_b + \hat V) \hat \Psi=0, 
\een
in the spacetime \eqref{Whitingds2}, where $\hat V$ is the potential 
\ben
\label{hatVKerr}
\hat V = \frac{s^2}{G\sqrt{\Delta} \sin \theta} \left( \frac{1-\cos \theta}{1+\cos \theta} + \frac{r-r_+}{r-r_-} \right) .
\een
The transformations are furthermore such that $\hat \Psi$ satisfies the boundary conditions of a QNM (in the new spacetime). 

Since $\hat V$ is real-valued, the  Klein-Gordon equation \eqref{hKG} follows from an action principle. This action principle 
yields a conserved Noether current (or, by Hodge duality, closed 3-form)
for each Killing field since they also Lie-derive $\hat V$. These currents can in fact be
 written in terms of the non-conserved canonical stress energy tensor,
\ben
\label{Tabdef}
\hat T_{ab} = \hat \nabla_a \hat \Psi^* \hat \nabla_b \hat \Psi - \frac{1}{2} \hat g_{ab}( \hat g^{cd} \hat \nabla_c \hat \Psi^* \hat \nabla_d \hat \Psi + \hat V |\hat \Psi|^2 )
\een
satisfying the dominant energy condition for $r \in (r_+,\infty)$  in view of $\hat V>0$. The Noether current 3-form conjugate to the KVF $X^a \in {\rm span}\{ T^a, \Phi^a \}$ is, 
\ben
\label{Jdef}
\hat J_{abc}[X] = \hat T_c{}^d X^c \hat \epsilon_{dabc}, 
\een
so that $\ud \hat {\boldsymbol J}[X] = 0$ when \eqref{hKG} holds. For $X^a = T^a$, the pull-back of $\hat {\boldsymbol J}[T]$ to any non-timelike subspace in a suitably oriented co-tangent space 
has a non-negative density by the dominant energy condition satisfied by $\hat T_{ab}$. This applies in particular to the co-tangent space of a constant $t$ hypersurface. 
From the QNM boundary conditions and the properties of Whiting's transformations, 
it follows that the integral 
\ben
\label{Eint}
E[\Sigma,X] = \int_\Sigma \hat {\boldsymbol J}[X]
\een
where $\Sigma$ is any constant $t$ surface (with this orientation understood) and where $X^a = T^a$, is (a) conserved and absolutely convergent, due to the decay of 
$\hat \Psi$ as $r \to r_+,\infty$ and $\theta \to 0, \pi$, and (b) coercive.
This is clearly incompatible with $\Imag(\omega)>0$ unless $\hat \Psi=0$, completing the proof by \cite{Whiting} 
of mode stability of Kerr. 

\subsection{Dual spacetime of Kerr-dS}

We now outline the the counterpart of these considerations in Kerr-dS, with the aim of providing a geometric understanding of the expressions 
\eqref{homegadef} and \eqref{hkappdef}, and of the bound \eqref{marcritabound} on real QNM frequencies. 
Consider a QNM solution $\Psi = R(r)S(\theta)e^{-i\omega t+im\phi}$ to the spin-$s$ Teukolsky master equation in Kerr-dS in our BL-like coordinates
\eqref{BLKdS}. In other words, at $r=r_+$, the radial solution is purely downgoing, and at $r=r_c$, it is purely outgoing; the precise meanings of these notions were recalled in 
\eqref{Rbcs} in sec. \ref{Heunradialsec}. In order to avoid awkward case distinctions in some formulas below, we assume without a loss of generality  in this and the following subsections that $m,s$ have the same sign; the other case can be reduced to this one by the transformation $\theta \to \pi -\theta$. 

We now:
\begin{itemize} 
\item change from $r$ to the coordinate $z$ \eqref{zr},
\item apply the mass symmetry transformation $R \to \hat R$ described in sec. \ref{sec:Rtransf} to the radial spin-$s$ Teukolsky master equation in the coordinate $z$, 
\item apply the differential transformation $S \to \hat S$ to the angular spin-$s$ Teukolsky master equation, described in \cite{Umetsu}, 
\item revert from the coordinate $z$ \eqref{zr} back to the coordinate $r$,
\item undo the separation of variables by recombining the transformed angular and radial functions into 
$\hat \Psi = \hat R(r) \hat S(\theta)e^{-i\omega t+im\phi}$.
\end{itemize}
Then, by elementary but very lengthy automated computations \cite{nb}, the new wave equation satisfied by $\hat \Psi$  can be re-written in 
covariant form \eqref{hKG}, for a Kerr-dS analogue of Whiting's metric $\hat g_{ab}$ \eqref{Whitingds2}, and a new potential $\hat V$. 
The line element of this ``dual'' metric that can be extracted is\footnote{Apparently our metric is different from what would be extracted 
from the wave equation given by \cite{Umetsu}. Since our mass symmetry transformations differ somewhat from the
integral transformations as used by \cite{Umetsu}, we have found it difficult to track the origin of this difference.}
\ben
\label{dsWKdS}
\ud\hat s^2 = \sqrt{\Delta_r (1+\alpha\cos^2 \theta)D} \sin \theta \left( \frac{
-H \ud t^2 - 2G\ud\phi \ud t + F\ud\phi^2 
}{D}+ \frac{\ud r^2}{\Delta_r} + \frac{\ud\theta^2}{1+\alpha\cos^2 \theta} 
\right).
\een
Here, $F,G,H$ are the functions of $(r,\theta)$, but not $(m,\omega)$ whose lengthy definitions are given
in appendix \ref{app:B} referring to \cite{nb}. In the case of $F,G$ these have a finite limit as $r_c \to \infty$, and reduce to those in Kerr \eqref{FGdef} in that limit, 
whereas  $H \to 0$. $-D$ is the determinant of the $\phi$-$t$ part, i.e. 
\ben
D=G^2+FH,
\een
which is given in appendix \ref{app:B} based on a very lengthy automated computation \cite{nb}.
For $r \in (r_+,r_c)$, we have that $D>0$ by lem. \ref{lem1}, 
expressing that $\hat g_{ab}$ has the Lorentzian signature in this region. The potential in the KG equation \eqref{hKG} becomes
\ben
\label{hV}
\hat V = \frac{P}{\sqrt{\Delta_r D (1+\alpha \cos^2 \theta)} \sin \theta} 
\een
in Kerr-dS, 
where $P$ is the function of $(r,\theta)$, but not $(m,\omega)$, given in appendix \ref{app:B}. By lem. \ref{lem2} there exists an $r'_c>r_c$ such that 
$\hat V > 0$ for all $r_o<r \le r_c'$. 

We now construct an extension of the spacetime beyond the original domain $r \in (r_+,r_c), \theta \in (0,\pi)$ in such a way that 
$\theta = 0, \pi$ represent time-like conformal boundaries, and $r=r_c, r_+$ represent bifurcate Killing horizons in the extension. 
For this, it is critical to understand the behaviors of $D(r,\theta), \Delta_r(r)$. By definition \eqref{Deltardef} $\Delta_r$ is a fourth order polynomial which 
has simple zeros at $r_o, r_-, r_+, r_c$. We show in lem. \ref{lem1} that there exists an $r_c'>r_c$ and 
a function $D'(r,\theta)$ such that $0 < D' < \infty$ for all $r \le r_c'$, and such that $D = -D'/[(r-r_c)(r-r_+)(r-r_-)]$.\footnote{
This is quite non-trivial e.g., because near either of the zeros $r_j, j \in \{-,+,c\}$, 
the formulas in \ref{app:B} show that  $F,G,H = \mathcal{O}((r-r_j)^{-1})$, so one would naively 
expect a different behavior for $D$ as $r \to r_j$.}
In particular, it follows from this that $\Delta_r D$ appearing under square-roots extends as a strictly positive smooth function 
to $r_o<r \le r_c'$.  

Formally, the norms of the Killing VFs $\hat K_j^a = T^a + \hat \omega_j \Phi^a$, where 
\ben
\label{hom}
\hat \omega_j = \frac{G}{F} \bigg|_{r=r_j}
\een
are such that $\hat g_{ab} \hat K_j^a  \hat K_j^b \to 0$ as $r \to r_j, j \in \{c,+,-\}$, suggesting that the coordinate locations $r=r_j, j \in \{c,+,-\}$ may 
represent horizons in an extension.  This is in fact the case. To construct this extension, we first introduce the $2\pi$-periodic ``corotating" coordinate
\ben
\hat \phi_j := \phi - \hat \omega_j t
\een
such that $\hat K_j^a = (\partial_t)^a$ in the new coordinates $(t,r,\theta,\hat \phi_j)$ near $r=r_j$ for each $j \in \{+,c,-\}$.
In a similar way as when introducing Kruskal-like coordinates for Kerr-dS, we next define a tortoise-like coordinate $\hat r_*$ by
\ben
\hat r_* = \frac{1}{2 \hat \kappa_+} \log |r - r_+| 
- \frac{1}{2 \hat \kappa_c} \log |r_c - r| 
 - \frac{1}{2 \hat \kappa_-} \log |r - r_-|.
\een
Here the modified surface gravities $\hat \kappa_j, j \in \{+,-,c\}$ are given by
\ben
\label{kappadef}
\hat \kappa_j = \frac{|\partial_r \Delta_r|}{2\sqrt{\Delta_r F}} \bigg|_{r=r_j}.
\een
Then we define the retarded and advanced coordinates
\ben
\hat u = t-\hat r_*, \quad 
\hat v = t+\hat r_*.
\een
To construct the extension across $r_+$, we define near $r=r_+$ 
the Kruskal-type coordinates 
\ben
\hat U_+ = -e^{-\hat \kappa_+ \hat u}, \quad \hat V_+ = e^{\hat \kappa_+ \hat v}. 
\een
Then we have $r-r_+ = -\hat U_+ \hat V_+(1+{\mathcal O}(\hat U_+ \hat V_+))$, 
$F/D, H/D, G/D = {\mathcal O}(1)$ with respect to the coordinates $\hat U_+, \hat V_+$ and, with the choice \eqref{kappadef}
of $\hat \kappa_+$, 
the metric $\hat g_{ab}$ is seen to possess an analytic continuation when expressed in our Kruskal-like coordinates across $r=r_+$. 
The Kruskal-like coordinates $x^\mu = (\hat U_+, \hat V_+, \hat \phi_+, \theta)$ cover the regions ${\rm I,II,III,IV}$ in fig. \ref{fig:1}.

A similar procedure can be 
applied to obtain an extension across  $r=r_c$ by defining 
\ben
\hat U_c = e^{\hat \kappa_c u},\quad  \hat V_c = - e^{-\hat \kappa_c \hat v}.
\een 
These regions are depicted as ${\rm II', IV', III'}$ in fig. \ref{fig:1}. The region $\rm III'$ is isometric to the region $\rm III$ and may thus be identified, as we have indicated in 
fig. \ref{fig:1}. Finally, suitably adapted procedure also works for $r=r_-$. 

In the Kruskal-like coordinates $(\hat V_j, \hat U_j, \theta, \hat \phi_j)$ the coordinate locations $r=r_j$ are bifurcate Killing horizons with surface gravities $\hat \kappa_j$ for each $j \in \{+,c, -\}$. In fact, the horizon Killing fields $\hat K_j^a = T^a + \hat \omega_j \Phi^a$ take the usual form in the Kruskal-like coordinates:
\ben
\label{hatKj}
\hat K^a_j \equiv T^a + \hat \omega_j \Phi^a = \hat \kappa_j \bigg(\hat U_j \partial_{\hat U_j} - \hat V_j \partial_{\hat V_j} \bigg)^a, \quad j \in \{c,+,-\},
\een
At $r = r_o$, $D$ remains non-zero, whereas $\Delta_r$ changes its sign. This means that the conformal factor in \eqref{dsWKdS} cannot be extended as 
a real function beyond $r=r_o$. In fact, one can see that $r=r_o$ represents a singularity. Likewise, one can see that the metric must have a singularity somewhere 
between $r_c'$ and $\infty$.

To summarize, we have the following theorem:
\begin{theorem} 
There is an analytic extension $(\hat {\mathcal M}, \hat g_{ab})$ of the dual spacetime metric \eqref{dsWKdS} to $r_o<r<r_c'$ for some $r_c'>r_c$.
The manifold structure is $\hat {\mathcal M} = \hat {\mathcal P} \times (0,\pi)_\theta \times S^1_\phi$, 
where $\hat {\mathcal P}$ is qualitatively given by the  Penrose diagrams fig. \ref{fig:1}.
This spacetimes has bifurcate Killing horizons at $r = r_j, j \in \{+, -, c\}$ whose surface gravities $\hat \kappa_j$ and angular velocities $\hat \omega_j$
are  \eqref{homegadef} and \eqref{hkappdef}. 

For a real QNM frequency, the transformed Teukolsky master equation for $\hat \Psi$ takes the form of a wave equation \eqref{hKG}
on $(\hat {\mathcal M}, \hat g_{ab})$ with a real-valued potential $\hat V$ given by \eqref{hV}. There exists an $r'_c>r_c$ such that 
$\hat V > 0$ for all $r \le r_c'$. In this domain, the stress tensor \eqref{Tabdef} therefore satisfies the dominant energy condition.
\end{theorem}

		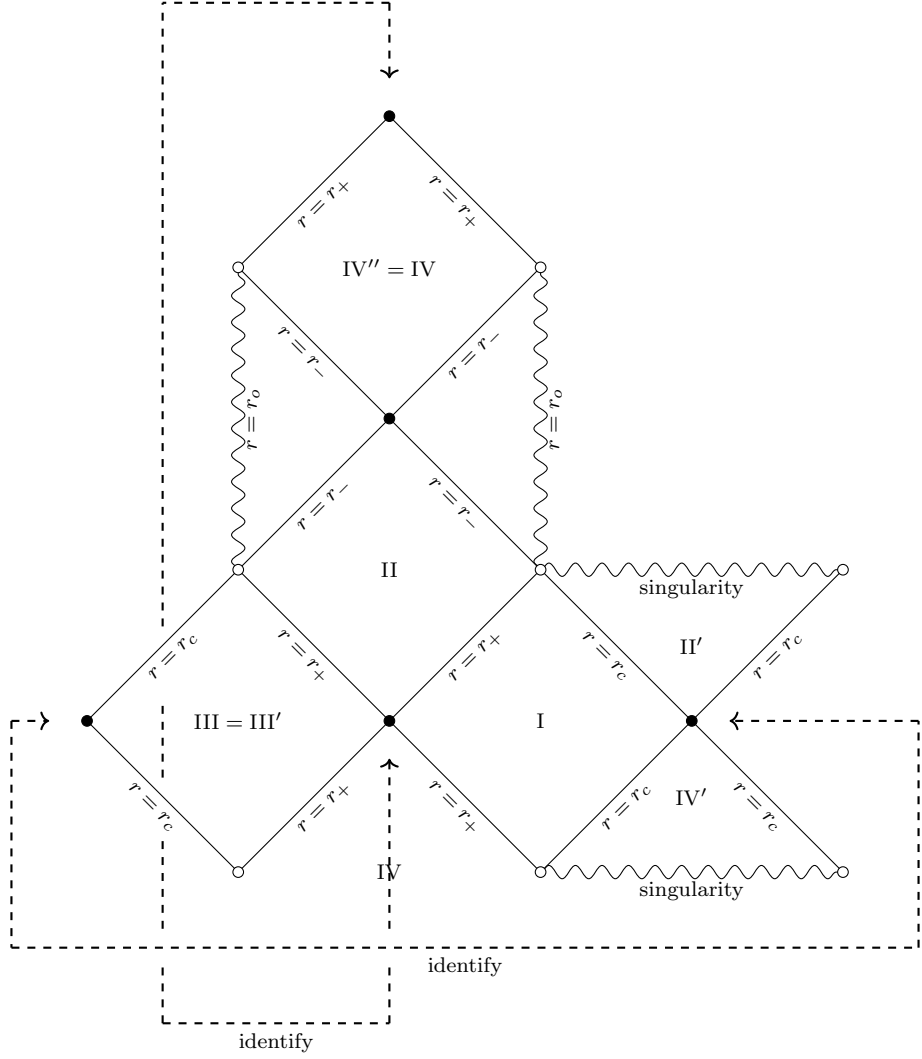
\begin{figure}
		\centering
			\begin{tikzpicture}[scale=1]
			\draw (0,4) -- (-2,6) node[midway, below, sloped]{\footnotesize $r=r_-$};
\draw (-2,6) -- (-4,8) node[midway, below, sloped]{\footnotesize $r=r_-$};	
\draw (-2,6) -- (0,8) node[midway, below, sloped]{\footnotesize $r=r_-$};
\draw (-2,6) -- (-4,4) node[midway, below, sloped]{\footnotesize $r=r_-$};		
\draw (0,8) -- (-2,10) node[midway, below, sloped]{\footnotesize $r=r_+$};
\draw (-2,10) -- (-4,8) node[midway, below, sloped]{\footnotesize $r=r_+$};		
\draw (0,0) -- (-2,2) node[midway, below, sloped]{\footnotesize $r=r_+$};
\draw (-2,2) -- (-4,4) node[midway, below, sloped]{\footnotesize $r=r_+$};
\draw (0,0) -- (2,2) node[midway, below, sloped]{\footnotesize $r=r_c$};
\draw (2,2) -- (0,4) node[midway, below, sloped]{\footnotesize $r=r_c$};
\draw (-2,2) -- (0,4) node[midway, below, sloped]{\footnotesize $r=r_+$};
\draw[black, thick, dashed,->] (-7,2) -- (-6.5,2);
\draw[black, thick,dashed, ->] (-2,11.5) -- (-2,10.5);
\draw[black, thick,dashed] (-2,11.5) -- (-5,11.5);

\draw[black, thick,dashed] (-5,-2) -- (-5,-1.25);
\draw[black, thick,dashed] (-5,-0.75) -- (-5,.5);
\draw[black, thick,dashed] (-5,1.25) -- (-5,2.25);
\draw[black, thick,dashed] (-5,3.25) -- (-5,11.5);

\draw[black, thick,dashed] (-5,-2) -- (-2,-2) node[midway, below, black]{\footnotesize identify};
\draw[black, thick,dashed] (-2,-2) -- (-2,-1.25);
\draw[black, thick,dashed,->] (-2,-.75) -- (-2,1.5);
\draw[black, thick,dashed] (-7,2) -- (-7,-1);
\draw[black, thick,dashed] (-7,-1) -- (5,-1) node[midway, below, black]{\footnotesize identify};
\draw[black, thick,dashed] (5,-1) -- (5,2);
\draw[black, thick,dashed,->] (5,2) -- (2.5,2);
\draw[snake it] (0,4) -- (0,8) node[midway, below, sloped]{\footnotesize $r=r_o$};
\draw[snake it] (-4,4) -- (-4,8) node[midway, below, sloped]{\footnotesize $r=r_o$};
\draw[snake it] (0,4) -- (4,4) node[midway, below, sloped]{\footnotesize singularity};
\draw[snake it] (0,0) -- (4,0) node[midway, below, sloped]{\footnotesize singularity};
\draw (2,2) -- (4,4) node[midway, below, sloped]{\footnotesize $r=r_c$};
\draw (2,2) -- (4,0) node[midway, below, sloped]{\footnotesize $r=r_c$};
\draw (-2,2) -- (-4,0) node[midway, below, sloped]{\footnotesize $r=r_+$};
\draw (-4,0) -- (-6,2) node[midway, below, sloped]{\footnotesize $r=r_c$};
\draw (-4,4) -- (-6,2) node[midway, below, sloped]{\footnotesize $r=r_c$};

\draw[fill=white] (-4,4) circle (2pt);
\draw[fill=white] (-4,0) circle (2pt);
\draw[fill=white] (4,4) circle (2pt);
\draw[fill=white] (4,0) circle (2pt);
\draw[fill=white] (0,8) circle (2pt);
\draw[fill=white] (-4,8) circle (2pt);
\draw[fill=white] (0,0) circle (2pt) node[below]{$$};
\draw[fill=black] (-2,2) circle (2pt);
\draw[fill=black] (-6,2) circle (2pt);
\draw[fill=black] (-2,6) circle (2pt);
\draw[fill=black] (-2,10) circle (2pt);
\draw[fill=white] (0,4) circle (2pt) node[above right]{$$};
\draw[fill=black] (2,2) circle (2pt);
\draw (0,2) node{\footnotesize ${\rm I}$};
\draw (-4,2) node{\footnotesize ${\rm III} = {\rm III}'$};
\draw (-2,4) node{\footnotesize ${\rm II}$};
\draw (-2,0) node{\footnotesize ${\rm IV}$};
\draw (2,3) node{\footnotesize ${\rm II}'$};
\draw (2,1) node{\footnotesize ${\rm IV}'$};
\draw (-2,8) node{\footnotesize ${\rm IV}''={\rm IV}$};
\end{tikzpicture}
\caption{The Penrose diagram $\hat {\mathcal P}$ of the dual spacetime.}
		\label{fig:1}
		\end{figure}

A lengthy computation using the expressions for $F,G$ in app. \ref{app:B} 
shows that the ``dual angular velocity'' parameters" \eqref{hom} are \cite{nb}\footnote{To compare these expressions, one has 
to use the definitions of the variables given in app. \ref{app:B}.}
\ben
\label{homegadef}
\begin{split}
\hat \omega_c =& 
\frac{2a}{2a^2+r_c(r_c+r_-) + r_+(r_c-r_-)}
= \frac{ \omega_c\kappa_c^{-1} - \omega_-\kappa_-^{-1}}{\kappa_c^{-1} - \kappa_-^{-1}}
,\\
\hat \omega_+=&
\frac{2a}{2a^2+r_c(r_+-r_-) + r_+(r_-+r_+)}
= \frac{ \omega_+\kappa_+^{-1} + \omega_-\kappa_-^{-1}}{\kappa_+^{-1} + \kappa_-^{-1}}.
\end{split}
\een
and ``dual surface gravities" \eqref{kappadef} are \cite{nb}
\ben
\label{hkappdef}
\begin{split}
\hat \kappa_c =& \frac{
(r_c-r_+)(r_+-r_o)(r_c-r_-)(r_--r_o)}{2L^2(1+\alpha)(r_++r_-)(2a^2-r_+r_--r_cr_o)} = (\kappa_c^{-1} - \kappa_-^{-1})^{-1},\\
\hat \kappa_+=& \frac{(r_c-r_+)(r_c-r_o)(r_+-r_-)(r_--r_o)
}{2L^2(1+\alpha)(r_c+r_-)(2a^2 -r_cr_--r_+r_o)} =
(\kappa_+^{-1} + \kappa_-^{-1})^{-1}, 
\end{split}
\een

\subsection{Partial mode stability of Kerr-dS}

We now reconsider the mode stability problem of the spin-$s$ Teukolsky equation from the geometric setup of the dual spacetime. For this, we consider the conserved current 3-form \eqref{Jdef} associated with a QNM $\hat \Psi$ on the region $\rm I$ between $r_+$ and $r_c$
of the dual spacetime $(\hat {\mathcal M}, \hat g_{ab})$, see fig. \ref{fig:1}. Then we form the conserved energy \eqref{Eint} 
for suitable surfaces $\hat \Sigma$, and for suitable Killing VFs $X^a$. For an explicit expression of the current 3-form \eqref{Jdef} when $X^a = T^a$, see appendix \ref{app:e}.

The radial mode functions in the dual spacetime (i.e., after the mass symmetry transformation, see sec. \ref{sec:Rtransf}) satisfy
\ben
\label{hRbc}
\hat R(r) = 
\begin{cases}
A_+ (r-r_+)^{\hat\eta_0/2}(1+\mathcal{O}(r-r_+)) & \text{as $r \to r_+$,}\\
A_c \, (r_c-r)^{\hat\eta_x/2} \, (1+\mathcal{O}(r-r_c)) & \text{as $r \to r_c$,}
\end{cases}
\een
where the parameters $\hat \eta_x, \hat \eta_0$ are given below in \eqref{MassSym}. The angular functions 
satisfy
\ben
\label{hSbc}
\hat S(\theta) = 
\begin{cases}
B_0 \theta^{|s|} (1+\mathcal{O}(\theta)) & \text{as $\theta \to 0$,}\\
B_\pi (1+\mathcal{O}(\pi-\theta)) & \text{as $\theta \to \pi$}
\end{cases}
\een
after a differential transformation described in \cite{Umetsu}.
Consequently, the Teukolsky scalar $\hat \Psi = \hat R(r) \hat S(\theta) e^{-i\omega t+im\phi}$ has the asymptotic behavior 
\ben
\label{hPsiasmpt}
\hat \Psi = 
\begin{cases}
A_+ \hat V_+^{-i(\omega - m\hat \omega_+)/\hat \kappa_+}(1 + \mathcal{O}(U_+V_+)) e^{im\hat \phi_+} & \text{as $r \to r_+$}\\
A_c \,\,  \hat U_c^{-i(\omega - m\hat \omega_c)/\hat \kappa_c} \,\, (1 + \, \mathcal{O}(U_cV_c) \, ) \, e^{im\hat \phi_c} & \text{as $r \to r_c$}
\end{cases}
\times
\begin{cases}
B_0 \theta^{|s|} (1+\mathcal{O}(\theta)) & \text{as $\theta \to 0$,}\\
B_\pi (1+\mathcal{O}(\pi-\theta)) & \text{as $\theta \to \pi$.}
\end{cases}
\een
in the Kruskal-like coordinates near the bifurcate Killing horizons corresponding to $r=r_+,r_c$, respectively.

\subsubsection{Constant $t$ slices} 
\label{part1}
At first, we consider a putative QNM such that $\Imag(\omega) >0$, i.e. mode stability would not hold. We take surfaces $\hat\Sigma_0, \hat \Sigma_1$ in region I  in fig. \ref{fig:5} to be defined by fixing the BL time coordinate at $t=t_0, t=t_1$ respectively. Then the energy integrals $E[T,\hat\Sigma_0]$ and $E[T,\hat\Sigma_1]$ as in \eqref{Eint} converge absolutely in view of 
\eqref{hSbc}, \eqref{hRbc}, \eqref{hV}, and e.g., the formulas for the current 3-form \eqref{Jdef} in appendix \ref{app:e}. In fact, by current conservation, $\ud \hat {\boldsymbol J}[T]=0$, 
we must have $E[T,\hat\Sigma_0] = E[T,\hat\Sigma_1]$. On the other hand, the fact that for a mode, ${\pounds}_T \hat \Psi = -i\omega \hat \Psi$, and the definition of the current 
\eqref{Jdef} give 
\ben
{\pounds}_T \hat {\boldsymbol J}[X]= \Imag(\omega) \hat {\boldsymbol J}[X], 
\een
where $\pounds$ is the Lie-derivative. This gives $E[T,\hat\Sigma_0] e^{\Imag(\omega)(t_1-t_0)}= E[T,\hat\Sigma_1]$, which
is possible if and only if $E[T,\hat\Sigma_0]=0$. 

		\begin{figure}
		\centering
			\begin{tikzpicture}[scale=1]
\draw (0,0) -- (-2,2) node[midway, below, sloped]{\footnotesize $\hat {\mathcal H}^-_+$};
\draw (0,0) -- (2,2) node[midway, below, sloped]{\footnotesize $\hat {\mathcal H}^-_c$};
\draw (2,2) -- (0,4) node[midway, above, sloped]{\footnotesize $\hat {\mathcal H}^+_c$};
\draw (-2,2) -- (0,4) node[midway, above, sloped]{\footnotesize $\hat {\mathcal H}^+_+$};
\draw[red, thick] (-2,2) to[out=-30,in=-150] (2,2);
\draw[red, thick] (-2,2) to[out=30,in=150] (2,2);
\draw (0,2.3) node{\footnotesize $\hat \Sigma_1$};
\draw (0,1.7) node{\footnotesize $\hat \Sigma_0$};
\draw[fill=white] (0,0) circle (2pt) node[below]{$$};
\draw[fill=black] (-2,2) circle (2pt);
\draw[fill=white] (0,4) circle (2pt) node[above right]{$$};
\draw[fill=black] (2,2) circle (2pt);
\end{tikzpicture}
\caption{The region $\rm I$ between $r_+$ and $r_c$ of $\hat {\mathcal P}$ in fig. \ref{fig:1} and the constant $t$ surfaces $\hat \Sigma_0, \hat\Sigma_1$.}
		\label{fig:5}
		\end{figure}
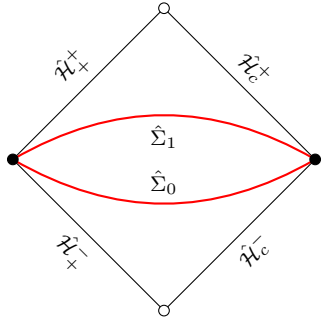

Unlike in Kerr, the condition that $E[T,\hat\Sigma_0]=0$ does not necessarily imply that $\hat \Psi = 0$ in Kerr-dS, 
because the integrand of $E[T,\hat\Sigma_0]$ is not sign-definite. This is a consequence of 
the fact that, while the stress tensor $\hat T_{ab}$ \eqref{Tabdef} satisfies the dominant energy condition, the VF $T^a$ has ergo regions (where it is spacelike). In fact, any other Killing VF of the dual metric has ergo regions, too, so neither can be used to infer $\hat \Psi=0$ by this kind of argument for generic Kerr-dS spacetimes.

On the other hand, $T^a$ is globally time-like in Schwarzschild-dS (the limiting case $a \to 0$), so the integrand of $E[T,\hat\Sigma_0]$ is positive definite, allowing us to conclude 
$\hat \Psi=0$ in this case. Thus, we recover the known result that there are no QNMs with $\Imag(\omega)>0$ of the spin-$s$ Teukolsky equation in Schwarzschild-dS spacetime.

\subsubsection{Hyperboloidal slices}
\label{part2}
Next, we consider a putative QNM such that $\Imag(\omega) =0$, i.e. mode stability would not hold on the real axis in Kerr-dS. Differently from before, 
 we now consider a hyperboloidal slice $\hat \Sigma_0$ intersecting the future horizons as in the following fig. \ref{fig:7}. We assume that the intersection 
 of $\hat \Sigma_0$ with the future horizons $\hat {\mathcal H}^+_c, \hat {\mathcal H}^+_+$ is a constant value of the respective Kruskal-like coordinates, $\hat V_+, \hat U_c$. 
 We also consider a second hyperboloidal slice $\hat \Sigma_1$ intersecting the future horizons as in the following fig. \ref{fig:7} to the future of $\hat \Sigma_0$, in such a 
 way that $\hat \Sigma_0$ is Lie-transported into $\hat \Sigma_1$ under the time-translation Killing VF $T^a$, i.e., $\hat \Sigma_1 = \mathscr{F}_t[\hat \Sigma_0]$, where $\mathscr{F}_t$
 is the flow generated by $T^a$. 
 
 We now consider a Killing VF $X^a = T^a +b\Phi^a$, where $b$ is an arbitrary real number chosen momentarily, and we consider the 
 current $\hat {\boldsymbol J}[X]$ defined in \eqref{Jdef} and \eqref{Tabdef}. By considering the Kruskal-like coordinate components of the dual metric
 $\hat g_{ab}$ \eqref{dsWKdS} near $\hat {\mathcal H}^+_+$, and near $\hat {\mathcal H}^+_c$, and by using the 
 asymptotic formulas \eqref{hPsiasmpt} for our QNM $\hat \Psi$ in Kruskal-like coordinates near $\hat {\mathcal H}^\pm_+$, and near $\hat {\mathcal H}^\pm_c$, we 
 can then easily see that the integrals \eqref{Eint} defining $E[X,\hat \Sigma_0]$ and $E[X, \hat\Sigma_1]$ are absolutely convergent.

		\begin{figure}
		\centering
			\begin{tikzpicture}[scale=1]
\draw (0,0) -- (-2,2) node[midway, below, sloped]{\footnotesize $\hat {\mathcal H}^-_+$};
\draw (0,0) -- (2,2) node[midway, below, sloped]{\footnotesize $\hat {\mathcal H}^-_c$};
\draw (2,2) -- (0,4) node[midway, above, sloped]{\footnotesize $\hat {\mathcal H}^+_c$};
\draw (-2,2) -- (0,4) node[midway, above, sloped]{\footnotesize $\hat {\mathcal H}^+_+$};
\draw[red, thick] (-1.5,2.5) to[out=-30,in=-150] (1.5,2.5);
\draw[red, thick] (-1.7,2.3) to[out=-30,in=-150] (1.7,2.3);
\draw (0,2.3) node{\footnotesize $\hat \Sigma_1$};
\draw (0,1.6) node{\footnotesize $\hat \Sigma_0$};
\draw[fill=white] (0,0) circle (2pt) node[below]{$$};
\draw[fill=black] (-2,2) circle (2pt);
\draw[fill=white] (0,4) circle (2pt) node[above right]{$$};
\draw[fill=black] (2,2) circle (2pt);
\end{tikzpicture}
\caption{The region $I$ between $r_+$ and $r_c$ of $\hat {\mathcal P}$ in fig. \ref{fig:1} and the hyperboloidal surfaces $\hat \Sigma_0, \hat\Sigma_1$.}
		\label{fig:7}
		\end{figure}
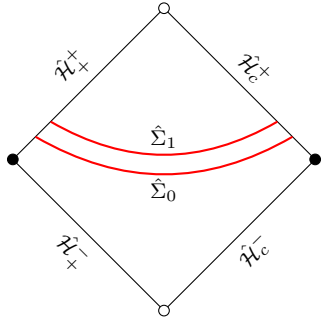

 Next, we integrate the conservation condition $\ud\hat {\boldsymbol J}[X]=0$ over the spacetime volume enclosed by $\hat\Sigma_1, \hat\Sigma_0$ within region $I$, and 
 we apply Gauss' theorem to write this condition in terms of certain boundary integrals. These boundary integrals consist of (a) 
 two boundary components of the volume within $\hat {\mathcal H}^+_+$ and $\hat {\mathcal H}^+_c$, and (b) two boundary components 
 associated with the boundary of $S^1_\phi \times (0,\pi)_\theta$ in the transverse directions, corresponding to $\theta = 0,\pi$ respectively, (c) 
 the two boundaries $\Sigma_0, \Sigma_1$. 
 
 The contributions from (a) are given by the integrated fluxes $\hat T_{ab} X^a \hat K_+^b|_{{\mathcal H}^+_+}$
 and $\hat T_{ab} X^a \hat K_c^b|_{{\mathcal H}^+_c}$, respectively. The contributions from (b) vanish due to the factors of $\sin \theta$ 
 in the formula for $\hat {\boldsymbol J}[X]$, see appendix \ref{app:e}, and the $\theta$-dependence in \eqref{hPsiasmpt}, which give an 
 additional factor vanishing at $\theta = 0$ in the case (happening only for $|s|\neq 0$) in which the potential $\hat V$ \eqref{hV} has a pole at $\theta = 0$.
 Finally, the contributions from (c) precisely cancel each other because, ${\pounds}_T \hat {\boldsymbol J}[X]= 0$ for a real QNM frequency.
 
 Computing the fluxes from case (a) using asymptotic formulas \eqref{hPsiasmpt} for our QNM $\hat \Psi$ as well as 
 Kruskal coordinate components of 
 $\hat g_{ab}$  near $\hat {\mathcal H}^+_+$ and  near $\hat {\mathcal H}^+_c$
 yields the condition:
 \ben
 |A_c|^2(\omega -bm)(\omega-m\hat \omega_c) +  |A_+|^2(\omega -bm)(\omega-m\hat \omega_+)=0.
 \een
We chose any $b$ so that we can cancel the $(\omega -bm)$-factor. Assuming momentarily that neither $A_+$ nor $A_c$ from the 
asymptotic formula vanish, it follows that $\omega-m\hat \omega_c$ and $\omega-m\hat \omega_+$ have an opposite sign. 
This implies at once that 
$
\hat \omega_+ < \omega/m < \hat \omega_c, 
$
which is precisely the bound for real QNM frequencies in Kerr-deSitter obtained by \cite{CTdC}.
Furthermore, as argued by \cite{CTdC}, if $A_+$ or $A_c$ vanish, then in fact $\hat \Psi = 0$. Thus, the dual geometry 
gives a geometric understanding of the bound by \cite{CTdC}.

\section{Conclusions and outlook}

In this paper, we have given an explanation of the spectral mass symmetries of the QNM spectrum of the spin-$s$
Teukolsky equation on Kerr-dS. Our explanation was based on a connection with the manifestly symmetric 
Painlevé VI equation via the isomonodromic flow. While this connection holds true quite generally, the special 
feature of the black hole problem is that the parameters of the equations are apparently non-generic:  
the potential in the Schrödinger form of the radial Teukolsky equation becomes {\em real} after the transformation. 

We have also exploited the ``dual geometry'' of Kerr-dS after the mass symmetry transformation. Also this dual 
geometry is special, in the sense that it has bifurcate Killing horizons. Furthermore, the potential in the wave equation 
after the symmetry transformation is real and positive definite. These features lead to a geometric interpretation of a 
bound due to \cite{CTdC} on possible QNM frequencies in Kerr-dS. 

Numerical investigations \cite{Irmert}, \cite{Yoshida:2010zzb} have so far not shown evidence of unstable QNMs 
in Kerr-dS though \cite{Irmert} has given evidence of unstable modes in a parameter range where the spacetime 
represents a naked singularity. In fact, a version of the isomonodromy approach can be used to establish 
the absence of unstable QNMs for Kerr-dS black holes with small $M/L$ (small cosmological constant) \cite{hollandsunpublished}, but 
we are not aware of a method capable of exploring the full subextremal parameter range. 

It is natural to ask whether the isomonodromy method has other applications in black hole perturbation theory. Such connections 
have in fact been considered already e.g., by \cite{daCunha:2015ana}, \cite{CarneirodaCunha:2015qln}, \cite{NMLC}, \cite{Bonelli:2021uvf}.
These works mostly aim at exploiting recent progress \cite{Gamayun:2012ma},
\cite{Gamayun:2013auu} related to series representations of the isomonodromic $\tau_{\rm VI}$-function---very closely related 
to the $\Sigma_{\rm VI}$-function used in this work---in order to 
find series expansions for scattering coefficients or related spectral data. 
It would be interesting to see if such methods can say something about the mode stability problem.

\medskip
{\bf Acknowledgement:} S.H. is grateful to Nagoya U. and SISSA for support and hospitality. He thanks G. Bonelli for discussions regarding \cite{Bonelli:2021uvf}.

\appendix

\section{Expressions for $F,G,H,D, P$}
\label{app:B}
We recall the notations $u=\cos \theta, \alpha = a^2/L^2$. In order to find $F$, $G$, $H$ in the dual metric \eqref{dsWKdS} we 
perform the mass symmetries on the radial Teukolsky equation, and the differential symmetries on the angular Teukolsky equations, 
as described in the main text. This results in a different radial equation, of the schematic form 
\ben
\biggl[\partial_r (\Delta_r \partial_r) + \frac{1}{\sin \theta} \partial_\theta \left( \sin \theta (1+\alpha \cos^2 \theta) \partial_\theta \right) + V_{\rm final}(r) \biggr] \hat \Psi = 0,
\een
where $\hat \Psi = \hat R(r) \hat S(\theta) e^{-i\omega t+im\phi}$. The potential $V_{\rm final}(r)$ can be found as {\tt Vr23final} in \cite{nb}, 
where one should identify $r_+=${\tt rp}, $r_-=${\tt rm}, $r_c=${\tt rc}, $r_o=${\tt ro}, $\alpha=${\tt al} $\omega=${\tt w}, $m=${\tt m}. 
From this potential, one then extracts all terms involving $\omega^2, m\omega, m^2$, and replaces these 
by the derivatives $-\partial_t^2, \partial_t \partial_\phi, -\partial_\phi^2$, respectively. 
The corresponding terms are found in \cite{nb} as {\tt Vr23tt}, {\tt Vr23pt}, {\tt Vr23pp}, respectively. 
Any terms not containing $m,\omega$ are collected in {\tt Vr230} in \cite{nb}.
Then one has a wave equation with potential {\tt Vr230} for $\hat \Psi$. 
This equation is similar, but not identical, to \cite[Eq. 5.15]{Umetsu}.
  
Covariantizing this equation as in  \eqref{hKG} is in principle a straightforward task involving the 
2-dimensional matrix {\tt g23} in \cite{nb}. The subsequent computations in \cite{nb} then deliver 
$F=${\tt Vr23tt}, $H=${\tt Vr23tt}, and $G=-(1/2)${\tt Vr23pt}, 
in the dual metric \eqref{dsWKdS}. For the dual surface gravities \eqref{kappadef}, 
the expressions {\tt kapc23}, {\tt kapp23} in \cite{nb}, which coincide with $\hat \kappa_c, \hat \kappa_+$ as in \eqref{hkappdef}
up to a sign disappearing when taking absolute values in  \eqref{kappadef}.

For $D=G^2+FH$, we find $-${\tt gdet23} in \cite{nb}, which is 
\begin{equation}
\label{Dinv}
\begin{aligned}
D^{-1} = &-\Biggl[
\alpha^2
(r-r_c)(r-r_-)(r-r_+)\\
&\times
\Bigl(
2r_c^{3}
+2r_-^{3}
+7r_-^{2}r_+
+7r_-r_+^{2}
+2r_+^{3}
+7r_c^{2}(r_-+r_+)
+r_c(
7r_-^{2}
+16r_-r_+
+7r_+^{2}
)
\Bigr)^{2}
(1+\alpha u^{2})
\Biggr]\\
&\times
\Biggl\{
(1+\alpha)^{3}
\Biggl[
r_c^{3}
+3r_c^{2}(r_-+r_+)
+(r_-+r_+)^{3}
+r_c(
3r_-^{2}
+7r_-r_+
+3r_+^{2}
)\\
&
-r\Bigl(
r_c^{2}
+r_-^{2}
+3r_-r_+
+r_+^{2}
+3r_c(r_-+r_+)
\Bigr)
\Biggr]\\
&\times
\Biggl[
a\,\alpha
\Bigl(
-r_c^{2}(r_-+r_+)
-r_-r_+(r_-+r_+)
+2r(r_c+r_-+r_+)^{2}
-r_c(r_-^{2}+4r_-r_++r_+^{2})
\Bigr)(u-1)\\
&+2a^{3}(r-r_o)
(1+\alpha u)
\Biggr]^{2}
\Biggr\}^{-1}.
\end{aligned}
\end{equation}
\begin{lem}
\label{lem1}
There exists $r_c'>r_c$ and $D'$ such that $D^{-1} = -(r-r_c)(r-r_-)(r-r_+) D^{\prime -1}$, where 
$\infty>D^{\prime -1} > 0$ for all $r\le r_c'$. 
\end{lem}
\begin{proof}
We use that $r_c>r_+>r_->0$, as assumed throughout this paper. Furthermore, we recall 
$\alpha = a^2/L^2$, where $L^2$ can be expressed by \eqref{Deltardef} as
\ben
\label{Lform}
L^2 =r_c^2 + r_+^2 + r_-^2 + r_+r_c + r_-r_c + r_+r_- + a^2.
\een
We also recall $u=\cos \theta$, so that $|u| \le 1$. We consider separately the 
three factors $[...]$ in the expression \eqref{Dinv}.

(i) The first $[...]$ is equal to $(r-r_c)(r-r_+)(r-r_-)$ times a manifestly positive number. 

(ii) By inspection, the term represented by the second $[...]$ in \eqref{Dinv} can be bounded from below by 
$2r_c(r_++r_-)^2>0$ when $r=r_c$. The $r$-independent term in $[...]$ is manifestly positive, 
therefore this $[...]$ must remain positive for all $r<r_c'$ up to some $r_c'>r_c$. 

(iii) By an elementary algebraic transformation, 
the term in the third $[...]$ in \eqref{Dinv} has the same sign as the following expression $A$:
\begin{equation}
\begin{aligned}
A &= (1-u) \Bigl[
r_c^2(r_++r_-) + r_c(r_++r_-)^2 + r_-r_+(r_c+r_++r_-) -2r(r_c+r_++r_-)^2
\Bigr]\\
&+2(1+\alpha u)(r+r_c+r_++r_-)(r_c^2 + r_+^2 + r_-^2 + r_+r_c + r_-r_c + r_+r_- + a^2).
\end{aligned}
\end{equation}
Obviously, if the first term in $[...]$ in $A$ is non-negative, then $A>0$ and we are done. 
Assume therefore that the $[...]$-terms in $A$ together are negative. Then $A$ is smallest if $u=-1$, 
i.e., we have the first inequality in 
\begin{equation}
\begin{aligned}
A &\ge 2 \Bigl[
r_c^2(r_++r_-) + r_c(r_++r_-)^2 + r_-r_+(r_c+r_++r_-) -2r(r_c+r_++r_-)^2
\Bigr]\\
&+2(1-\alpha)(r+r_c+r_++r_-)(r_c^2 + r_+^2 + r_-^2 + r_+r_c + r_-r_c + r_+r_- + a^2)\\
&= 2 \Bigl[
r_c^2(r_++r_-) + r_c(r_++r_-)^2 + r_-r_+(r_c+r_++r_-) -2r(r_c+r_++r_-)^2
\Bigr]\\
&+2(r+r_c+r_++r_-)(r_c^2 + r_+^2 + r_-^2 + r_+r_c + r_-r_c + r_+r_-)\\
&\ge 2 \Bigl[
r_c^2(r_++r_-) + r_c(r_++r_-)^2 + r_-r_+(r_c+r_++r_-)\\
&+(r_c+r)(r_c+r_++r_-)^2-2r(r_c+r_++r_-)^2 -2r_c(r_cr_++r_cr_-+r_+r_-)\\
&+(r_++r_-)(r_c^2 + r_+^2 + r_-^2 + r_+r_c + r_-r_c + r_+r_-)\Bigr]\\
&= 2 \Bigl[
r_-r_+(r_++r_-) + r_c(r_++r_-)^2 +r_+r_c^2\\
&+(r_++r_-)(r_+^2 + r_-^2  + r_-r_c + r_+r_-)+(r_c-r)(r_c+r_++r_-)^2\Bigr] >0.
\end{aligned}
\end{equation}
In the second line, we used $\alpha=a^2/L^2$ and \eqref{Lform}. In the last step, we assumed that $r \le r_c$. 
So $A>0$ if $r \le r_c$, which is an open condition and hence remains valid for $r \le r_c$ for some $r'_c>r_c$.
\end{proof}

The quantity $P$ appearing in the potential $\hat V$ \eqref{hV} corresponds to $-${\tt Vr230} in \cite{nb}, since it 
is that part of $V_{\rm final}$ not containing $m^2, m\omega, \omega^2$, and hence no derivatives after undoing 
the mode decomposition. Using \eqref{Lform}, $-${\tt Vr230} in \cite{nb} may be simplified further to
\begin{equation}
\label{Potinv}
\begin{aligned}
P=&
s^2(1+\alpha) \frac{1+u}{1-u}+\frac{\alpha}{a^2 (r+r_++r_-+r_c)} \Biggl[
2 r^2 (r+r_c+r_++r_-)\\
&+s^2 \Bigl(
r_c^3+r_+^3+r_-^3+ 3r_c^2 r_- + 3r_c^2 r_+ + 3r_-^2 r_c + 3r_+^2r_c + 3 r_-^2r_+ + 3 r_+^2r_- + 7r_cr_+r_-
\Bigr)\\
&-s^2r \Bigr(
r_c^2 + r_+^2 + r_-^2 + 3r_cr_- + 3r_cr_+ + 3r_-r_+
\Bigl) 
\Biggr]+2\alpha u^2 .
\end{aligned}
\end{equation}
\begin{lem}
\label{lem2}
There exists $r_c'>r_c$ such that $P \ge 0$ for all $r_o < r\le r_c'$. 
\end{lem}
\begin{proof}
Since $u=\cos \theta$, we have $1\pm u \ge 0$. $P$ will thus be less than or equal to the expression 
obtained by putting $r=r_c$ in the last line of  \eqref{Potinv}, since this is the only term which is not sign definite. We thereby have, for 
$r>r_o=-r_+-r_--r_c$, that
\begin{equation}
\label{Potinv1}
\begin{aligned}
P\ge& 
\frac{\alpha s^2}{a^2 (r+r_++r_-+r_c)}
 \Bigl(
r_+^3+r_-^3+2r_+^2r_c +2r_-^2r_c + 3 r_-^2r_+ + 3 r_+^2r_- + 4r_cr_+r_-
\Bigr) > 0
\end{aligned}
\end{equation}
by dropping terms that are manifestly non-negative for all $r$. Therefore, we have $P>0$ for $r_o < r \le r_c$. 
Since this is an open condition, the statement follows immediately. 
\end{proof}
\section{Coordinate expression of Noether current $\hat {\boldsymbol J}[T]$}
\label{app:e}

Using the notations of appendix \ref{app:B}, we have
\ben
\begin{split}
&\hat {\boldsymbol J}[T] =2(1+\alpha \cos^2 \theta) \Real \bigg(\partial_\theta \hat \Psi^* \partial_t \hat \Psi \bigg) \,  \sin \theta \, \ud t \wedge \ud\phi \wedge \ud r
-2\Delta_r(r) \Real\bigg(\partial_r \hat \Psi^* \partial_t \hat \Psi\bigg) \,  \sin \theta \, \ud t \wedge \ud\phi \wedge \ud\theta
\\
& +\bigg\{
F(r,\theta) |\partial_t \hat \Psi|^2 + H(r,\theta) |\partial_\phi \hat \Psi|^2 + \Delta_r(r) |\partial_r \hat \Psi|^2 + (1+\alpha \cos^2 \theta) |\partial_\theta \hat \Psi|^2  + P(r,\theta)  |\hat \Psi|^2 \bigg\}  \sin \theta \, \ud\theta \wedge \ud\phi \wedge \ud r\\
&-2\bigg[ -2G(r,\theta) |\partial_t \hat \Psi|^2 + H(r,\theta) \Real\bigg(\partial_\phi \hat \Psi^* \partial_t \hat \Psi \bigg) 
\bigg]  \sin \theta \, \ud t \wedge \ud\theta \wedge \ud r .
\end{split}
\een

\section{Geometry of Whiting's metric}
\label{app:A}
Whiting's metric \eqref{Whitingds2} associated with Kerr has the following properties.
\begin{enumerate}
\item $T^a = (\partial_t)^a$ is a null Killing VF, $\Phi^a = (\partial_\phi)^a$ is a Kiling VF with $2\pi$-periodic orbits.
\item The Einstein tensor has the form of a null fluid in the direction $T^a$, i.e. $\hat G_{ab} \propto T_a T_b$; in particular $\hat R=0$.
\item  $\hat C_{abcd}T^d = 0$, that is $\hat g_{ab}$ is of type N in the Petrov Weyl tensor classification (see e.g., \cite{waldbook}), with three times repeated principal null direction 
given by $T^a$.\footnote{In fact, the only non-zero Weyl components are $\hat C_{r\phi r\phi}, \hat C_{r\phi\theta\phi}, \hat C_{\theta\phi\theta\phi}$.
Note that Kerr is of type D in the Petrov Weyl tensor classification, and that it thereby has further non-zero Weyl components.}
\item By the Goldberg-Sachs theorem $T^a$ is tangent to a congruence of shear free null geodesics.
\item Since $T^a$ is tangent and normal to the surfaces given by $\phi=$ const. we have $T_{[a} \hat \nabla_b T_{c]} = 0$, i.e., $T^a$ is twist free.
\item Since $T^a$ is Killing, it follows that $\hat \nabla_a T^a = 0$, implying that $T^a$ is expansion free because it is geodesic by 4).
\item Using 6), $\hat \nabla^a \hat G_{ab} = 0$ and item 2) we get that $T^a \hat \nabla_a T^b = 0$. 
\item By Prop. 1 of \cite{Lewand}, there exists a spinor field $\tau^A$ such that 
$T^a = \tau^A \bar \tau^{A'}$ and such that the twistor spinor equation 
\ben
(\hat \nabla^{(A}{}_{(A'} \bar \tau_{B')})\tau^{B)} = 0
\een
holds.
\end{enumerate}
The twistor spinor $\tau^A$ may be used to construct a (complex) conformal Killing-Yano tensor via 
\ben
\hat Y_{ab}=\bar \epsilon_{A'B'} \tau_A \tau_B.
\een
As usual, this gives rise to the (complex) conformal Killing tensor $\hat K_{ab} = \hat Y_a{}^c\hat Y_{cb}$. Its real and imaginary parts give rise to constants of motion for null geodesics.


\begin{thebibliography}{10}

\bibitem{Akcay:2010vt}
S.~Akcay and R.~A.~Matzner,
``Kerr-de Sitter Universe,''
Class. Quant. Grav. \textbf{28}, 085012 (2011)

\bibitem{Andersson:2016epf}
L.~Andersson, S.~Ma, C.~Paganini and B.~F.~Whiting,
``Mode stability on the real axis,''
J. Math. Phys. \textbf{58}, 072501 (2017)

\bibitem{Bonelli:2021uvf}
G.~Bonelli, C.~Iossa, D.~P.~Lichtig and A.~Tanzini,
``Exact solution of Kerr black hole perturbations via CFT2 and instanton counting: Greybody factor, quasinormal modes, and Love numbers,''
Phys. Rev. D \textbf{105}, 044047 (2022) 

\bibitem{Borthwick:2018qsb}
J.~Borthwick,
``Maximal Kerr{\textendash}de Sitter spacetimes,''
Class. Quant. Grav. \textbf{35}, 215006 (2018) 
[erratum: Class. Quant. Grav. \textbf{39}, 219501 (2022) ]



\bibitem{daCunha:2015ana}
B.~Carneiro da Cunha and F.~Novaes,
``Kerr Scattering Coefficients via Isomonodromy,''
JHEP \textbf{11}, 144 (2015)


\bibitem{CarneirodaCunha:2015qln}
B.~Carneiro da Cunha and F.~Novaes,
``Kerr{\textendash}de Sitter greybody factors via isomonodromy,''
Phys. Rev. D \textbf{93}, 024045 (2016) 

\bibitem{Carter}
B. Carter, in {\it Les Astres Occlus,} ed. by B. DeWitt, C. M. DeWitt, (Gordon and Breach, New York, 1973)

\bibitem{CTdC}
M.~Casals and R.~T.~da Costa,
``Hidden Spectral Symmetries and Mode Stability of Subextremal Kerr(-de Sitter) Black Holes,''
Commun. Math. Phys. \textbf{394}, no.2, 797-832 (2022)

\bibitem{ChambersMoss}
C.~M.~Chambers and I.~G.~Moss,
``Stability of the Cauchy horizon in Kerr-de Sitter space-times,''
Class. Quant. Grav. \textbf{11}, 1035-1054 (1994)

\bibitem{Chrusciel:2012jk}
P.~T.~Chrusciel, J.~Lopes Costa and M.~Heusler,
``Stationary Black Holes: Uniqueness and Beyond,''
Living Rev. Rel. \textbf{15}, 7 (2012)


\bibitem{Dolan:2007mj}
S.~R.~Dolan,
``Instability of the massive Klein-Gordon field on the Kerr spacetime,''
Phys. Rev. D \textbf{76}, 084001 (2007)

\bibitem{Dold:2015cqa}
D.~Dold,
``Unstable Mode Solutions to the Klein\textendash{}Gordon Equation in Kerr-anti-de Sitter Spacetimes,''
Commun. Math. Phys. \textbf{350}, no.2, 639-697 (2017)

\bibitem{Dyatlov:2010hq}
S.~Dyatlov,
``Quasi-normal modes and exponential energy decay for the Kerr-de Sitter black hole,''
Commun. Math. Phys. \textbf{306}, 119-163 (2011)

\bibitem{Gamayun:2012ma}
O.~Gamayun, N.~Iorgov and O.~Lisovyy,
``Conformal field theory of Painlev\'e VI,''
JHEP \textbf{10}, 038 (2012)

\bibitem{Gamayun:2013auu}
O.~Gamayun, N.~Iorgov and O.~Lisovyy,
``How instanton combinatorics solves Painlev\'e VI, V and IIIs,''
J. Phys. A \textbf{46}, 335203 (2013)

\bibitem{Graf:2022fve}
O.~Graf and G.~Holzegel,
``Mode stability results for the Teukolsky equations on Kerr-anti-de Sitter spacetimes,''
Class. Quant. Grav. \textbf{40}, 045003 (2023)


\bibitem{Green:2015kur}
S.~R.~Green, S.~Hollands, A.~Ishibashi and R.~M.~Wald,
``Superradiant instabilities of asymptotically anti-de Sitter black holes,''
Class. Quant. Grav. \textbf{33}, 125022 (2016) 



\bibitem{Hatsuda:2020sbn}
Y.~Hatsuda,
``Quasinormal modes of Kerr-de Sitter black holes via the Heun function,''
Class. Quant. Grav. \textbf{38}, 025015 (2020) 

\bibitem{Hintz}
P.~Hintz,
``Mode stability and shallow quasinormal modes of Kerr-de Sitter black holes away from extremality,''
[arXiv:2112.14431 [gr-qc]].

\bibitem{Hintz:2016gwb}
P.~Hintz and A.~Vasy,
``The global non-linear stability of the Kerr-de Sitter family of black holes,''
doi:10.4310/acta.2018.v220.n1.a1
[arXiv:1606.04014 [math.DG]].

\bibitem{Hintz:2022lgf}
P.~Hintz and G.~Holzegel,
``Recent progress in~general relativity,''
doi:10.4171/icm2022/128



\bibitem{hollandsunpublished}
S. Hollands, unpublished notes

\bibitem{Irmert}
L. Irmert, {\it Numerical survey
on the stability of Kerr de-Sitter
space-times,} B.Sc. thesis, U. Leipzig (2025)

\bibitem{Jimbo:1981tov}
M.~Jimbo, T.~Miwa and K.~Ueno,
``Monodromy preserving deformation of linear ordinary differential equations with rational coefficients: I. General theory and {\ensuremath{\tau}}-function,''
Physica D \textbf{2}, 306-352 (1981) 

\bibitem{JimboMiwaII}
M.~Jimbo and T.~Miwa,
``Monodromy Preserving Deformations Of Linear Differential Equations With Rational Coefficients. 2.,''
Physica D \textbf{2}, 407-448 (1981)

\bibitem{Jimbo}
M.~Jimbo,
``Monodromy problem and the boundary condition for some Painlev\'e equations,''
Publ.\ RIMS \textbf{18}, 1137-1161 (1982)

\bibitem{Khanal:1983vb}
U.~Khanal,
``Rotating Black Hole in Asymptotic deSitter Space: Perturbation of the Space-Time with Spin Fields'',
Phys. Rev. D \textbf{28}, 1291-1297 (1983)


\bibitem{Leaver:1985ax}
E.~W.~Leaver,
``An Analytic representation for the quasi normal modes of Kerr black holes,''
Proc. Roy. Soc. Lond. A \textbf{402}, 285-298 (1985)

\bibitem{Lewand}
J. Lewandowski, ``Twistor equation in curved spacetime,'' Class. Quantum Grav. 8, L11 (1991) 

\bibitem{NMLC}
F.~Novaes, C.~Marinho, M.~Lencs\'es and M.~Casals,
``Kerr-de Sitter Quasinormal Modes via Accessory Parameter Expansion,''
JHEP \textbf{05}, 033 (2019)

\bibitem{Okamoto}
K. Okamoto, ``Studies on the Painlevé equations,'' 
Annali di Matematica pura ed applicata 146, 337–381 (1986). 

\bibitem{Press:1973zz}
W.~H.~Press and S.~A.~Teukolsky,
``Perturbations of a Rotating Black Hole. II. Dynamical Stability of the Kerr Metric,''
Astrophys. J. \textbf{185}, 649-674 (1973)

\bibitem{Regge:1957td}
T.~Regge and J.~A.~Wheeler,
``Stability of a Schwarzschild singularity,''
Phys. Rev. \textbf{108}, 1063-1069 (1957)

\bibitem{Ronveaux}
A.~Ronveux (ed.), {\it Heun's differential equation,} Oxford University Press (1995).

\bibitem{Shlapentokh-Rothman:2013hza}
Y.~Shlapentokh-Rothman,
``Quantitative Mode Stability for the Wave Equation on the Kerr Spacetime,''
Annales Henri Poincare \textbf{16}, 289-345 (2015)


\bibitem{Shlapentokh-Rothman:2020vpj}
Y.~Shlapentokh-Rothman and R.~Teixeira da Costa,
``Boundedness and decay for the Teukolsky equation on Kerr in the full subextremal range $|a|<M$: frequency space analysis,''
[arXiv:2007.07211 [gr-qc]].

\bibitem{Shlapentokh-Rothman:2013ysa}
Y.~Shlapentokh-Rothman,
``Exponentially growing finite energy solutions for the Klein-Gordon equation on sub-extremal Kerr spacetimes,''
Commun. Math. Phys. \textbf{329}, 859-891 (2014)

\bibitem{STU98}
H.~Suzuki, E.~Takasugi and H.~Umetsu,
``Perturbations of Kerr-de Sitter black hole and Heun's equations,''
Prog. Theor. Phys. \textbf{100}, 491-505 (1998)

\bibitem{Suzuki:1999nn}
H.~Suzuki, E.~Takasugi and H.~Umetsu,
``Analytic solutions of Teukolsky equation in Kerr-de Sitter and Kerr-Newman-de Sitter geometries,''
Prog. Theor. Phys. \textbf{102}, 253-272 (1999)

\bibitem{TeixeiradaCosta:2019skg}
R.~Teixeira da Costa,
``Mode stability for the Teukolsky equation on extremal and subextremal Kerr spacetimes,''
Commun. Math. Phys. \textbf{378}, 705-781 (2020)

\bibitem{Teukolsky:1973ha}
S.~A.~Teukolsky,
``Perturbations of a rotating black hole. 1. Fundamental equations for gravitational electromagnetic and neutrino field perturbations,''
Astrophys. J. \textbf{185}, 635-647 (1973)

\bibitem{Umetsu}
H.~Umetsu,
``A Conserved energy integral for perturbation equations in the Kerr-de Sitter geometry,''
Prog. Theor. Phys. \textbf{104}, 743-755 (2000)


\bibitem{waldbook}
R. M. Wald, {\it General Relativity,} U. Chicago Press (1984)

\bibitem{Whiting}
B.~F.~Whiting,
``Mode Stability of the Kerr Black Hole,''
J. Math. Phys. \textbf{30}, 1301 (1989)

\bibitem{nb}
Notebook {\tt KdSPaper.nb}, enabled by 
Wolfram Research, Inc., Mathematica, Version 14.3, Champaign, IL (2025)

\bibitem{Yoshida:2010zzb}
S.~Yoshida, N.~Uchikata and T.~Futamase,
``Quasinormal modes of Kerr-de Sitter black holes,''
Phys. Rev. D \textbf{81}, 044005 (2010)

\bibitem{Zerilli:1970se}
F.~J.~Zerilli,
``Effective potential for even parity Regge-Wheeler gravitational perturbation equations,''
Phys. Rev. Lett. \textbf{24}, 737-738 (1970)






\end{thebibliography}
\end{document}